\pdfoutput=1
\documentclass[12pt, a4paper]{article}
\usepackage[T1]{fontenc}

\usepackage[vmargin = 1.2in, hmargin =1.2in]{geometry}
\usepackage{setspace}
\usepackage[
activate={true,nocompatibility},
factor=100,
expansion=false,           
protrusion=true,          
tracking=true,            
spacing=true,             
shrink=7,                
stretch=7,               
kerning=true,             
final                      
]{microtype}
\microtypecontext{spacing=nonfrench}

\usepackage{mathtools, amsthm, amssymb, bm}
\mathtoolsset{centercolon}
\allowdisplaybreaks

\usepackage[caption=false]{subfig}
\usepackage{caption}
\usepackage{graphicx,tikz,xcolor}
\usetikzlibrary{arrows.meta,positioning,patterns,decorations.pathreplacing,calligraphy}
\usepackage{pgfplots, pgfplotstable}
\pgfplotsset{compat=newest}
\pgfplotsset{
	label style={font=\scriptsize},  
	legend style={
		font=\footnotesize,
		nodes={font=\footnotesize}
	},
	ticklabel style={font=\scriptsize},
	axis lines=middle,
	axis line style={-latex},
	ylabel style={
		above,
		xshift=5pt,
		font=\scriptsize
	},
	xlabel style={
		below,
		xshift=-6pt,
		yshift=-12pt,
		font=\scriptsize
	},
	legend pos=north west,
	width=0.38\textwidth, scale only axis,
	samples = 700,
}
\usepgfplotslibrary{fillbetween}
\allowdisplaybreaks[1]
\pgfmathsetseed{5}
\pgfplotstablenew[
create on use/randwalk1/.style={
	create col/expr accum={\pgfmathaccuma +0.01+ 0.1*rand}{0}
},
create on use/randwalk2/.style={
	create col/expr accum={\pgfmathaccuma -0.004 + 0.12*rand}{0}
},
create on use/randwalk3/.style={
	create col/expr accum={\pgfmathaccuma - 0.005+ 0.1*rand}{0}
},
columns={randwalk1,randwalk2,randwalk3}
]
{600}
\loadedtable

\definecolor{blue}{RGB}{43,147,206}
\definecolor{GmailBlue}{RGB}{42, 93, 176}

\definecolor{red}{RGB}{221,126,0}
\definecolor{orange}{RGB}{208,106,11} 

\definecolor{green}{RGB}{0,158,115} 
\definecolor{gray}{RGB}{73, 73, 73}
\definecolor{yellow}{RGB}{230,168,0}
\definecolor{pink}{RGB}{211,0,214}

\usepackage[authoryear]{natbib}
\usepackage[
bookmarks=true,
bookmarksnumbered=true,
bookmarksopen=true,
pdfborder={0 0 0},
breaklinks=true,
colorlinks=true,
linkcolor=GmailBlue,
citecolor=GmailBlue,
filecolor=GmailBlue,
urlcolor=GmailBlue,
hypertexnames=false,
plainpages=false,
]{hyperref}

\makeatletter
\newcommand\org@hypertarget{}
\let\org@hypertarget\hypertarget
\renewcommand\hypertarget[2]{%
	\Hy@raisedlink{\org@hypertarget{#1}{}}#2%
} \makeatother

\makeatletter
\renewcommand*{\NAT@spacechar}{\ }
\makeatother 
\makeatletter
\AtBeginDocument{ \hypersetup{ 
		pdftitle = {Coordination in complex environments}, 
		pdfauthor = {Pietro Dall'Ara} 
} }
\makeatother

\usepackage{apptools}
\AtAppendix{\counterwithin{lemma}{section}}
\AtAppendix{\counterwithin{proposition}{section}}
\AtAppendix{\counterwithin{theorem}{section}}
\AtAppendix{\counterwithin{corollary}{section}}
\AtAppendix{\counterwithin{remark}{section}}
\AtAppendix{\counterwithin{definition}{section}}

\theoremstyle{plain}

\newtheorem{proposition}{Proposition}
\newtheorem{lemma}{Lemma}
\newtheorem{corollary}{Corollary}

\theoremstyle{definition}

\newtheorem{remark}{Remark}

\newcommand*\diff{\mathop{}\!\mathrm{d}}

\DeclareMathOperator*{\argmax}{Argmax}

\makeatletter
\let\originalleft\left
\let\originalright\right
\renewcommand{\left}{\mathopen{}\mathclose\bgroup\originalleft}
\renewcommand{\right}{\aftergroup\egroup\originalright}
\renewcommand*{\NAT@spacechar}{\ }
\makeatother

\DeclareMathOperator*{\var}{Var}
\DeclareMathOperator*{\cov}{Cov}

\author{\href{https://sites.google.com/view/pietrodallara12}{\color{black}{Pietro Dall'Ara}}}

\usepackage{datetime}
\newdateformat{specialdate}{\THEDAY~\monthname[\THEMONTH] \THEYEAR}
\date{\specialdate\today}

\title{\bf Coordination in complex environments\footnote{This paper is based on the first chapter of my PhD dissertation at BC. I am indebted to Mehmet Ekmekci for continual guidance, and to Laurent Mathevet, Utku \"Unver, and Bumin Yenmez for their advice. I am grateful for comments from Zeinab Aboutalebi, Georgy Artemov, Matteo Bizzarri, Giacomo Calzolari, Krishna Dasaratha, Marco Errico, Marcelo Fernandez, Joel Flynn, Hideo Konishi, Maciej Kotowski, Alessandro Lavia, Lucas Maestri, Chiara Margaria, Andrea Mattozzi, Konuray Mutluer, Philip Neary, Parth Parihar, Luigi Pollio, Sven Rady, Giacomo Rubbini, Tom Rutter, Junze Sun, and Yu Fu Wong; contact: \href{mailto:pietro.dallara@unina.it}{pietro.dallara@unina.it}.}}

\begin{document}

	\maketitle
		\begin{abstract}
		 Coordination is an important aspect of innovative contexts, where: the more innovative a course of action, the more uncertain its outcome. To study the interplay of coordination and informational ``complexity'', I embed a beauty-contest game into a complex environment. I identify a new conformity phenomenon. This effect may push towards exploration of unknown alternatives, or constitute a status quo bias, depending on the network structure of the players' interaction. In an application, I show that an organization with decentralized authority can implement profit maximization in a sufficiently complex environment.
	\end{abstract}

	
	\newpage 
	\tableofcontents
	\newpage 
	\section{Introduction}
	Coordination poses challenges in highly uncertain environments. Consider retailers that share the same manufacturer and choose marketing campaigns, as in co-op advertising \citep{jorgensen_survey_2014}. Innovative advertisement comes with uncertainty about the brand image of the manufacturer. Moreover, retailers need to coordinate their advertisements as well as succeed in distinct markets. Does uncertainty lead to a unified brand image, and do the campaigns align with the manufacturer's interests? Coordination is also an important aspect of technological  innovation. For example, hospitals find it advantageous to select  popular electronic medical record vendors \citep{lin_strategic_2023}. Does uncertainty lead to innovation?

	This paper studies coordination problems in the face of ``incremental'' uncertainty, referred to as \emph{complexity}, such that: the more innovative a decision, the more uncertain its outcome.
	
	
	In the model, every player wants the outcome of her action to be close to a target. The target of a player combines her fixed favorite outcome with the outcomes of the other players, leading to a coordination-adaptation tradeoff. A given network of players determines how a target weighs individual outcomes. Analogous coordination motives arise in  organizations, oligopolies, and labor markets \citep{marschak_economic_1972, topkis_supermodularity_1998, diamond_aggregate_1982}.
	
	The complexity is modeled as uncertainty about how individual actions are mapped to individual outcomes. There is an outcome function determining the outcome of every single policy,  and players independently choose policies. Players know that the outcome function is the realized path of a Brownian motion. The initial point of the Brownian motion represents the \emph{status quo}: a known outcome corresponds to the ``initial'' policy. Instead, different policies than the status quo lead to outcome that are known only up to a noise. The more an outcome differs in expectation from the status-quo outcome, the higher its variance; this approach to modeling complex environments is introduced by \citet{callander_searching_2011}.

	This approach captures the idea that more innovative decisions lead to more volatile outcomes. The uncertainty is summarized by a status quo and a \emph{covariance structure}. The status quo is an action inducing relatively low uncertainty in the corresponding individual outcome. The covariance structure describes the likelihood that two actions yield similar outcomes. Complexity is relevant when deciding about the adoption of novel pricing strategies and how boldly to innovate in new technologies. In these examples, the covariance structure is important due to the strategic interaction ensuing from coordination motives, and the status quo is, respectively, a tried-and-tested pricing rule, and a technology currently in place.

	The interplay of coordination and complexity leads to a novel conformity phenomenon. In particular, expected outcomes are closer across players than in an environment
	without complexity.\footnote{This observation is described in Section \ref{sec:twoplayers}.} This conformity occurs in addition to the status-quo bias identified by \citeauthor{callander_searching_2011} and the conformity merely due to coordination motives. The equilibrium characterization separates the new conformity from previously	studied phenomena, by decomposing expected outcomes in terms of two elements: the equilibrium outcomes in a non-complex environment, and a new effect arising from the interplay between complexity and coordination (Proposition \ref{prop:equilibrium}).

	The new element in the equilibrium characterization arises from an endogenous
	leader-follower relationship among players.
	Consider the two ways in which the policy of a player influences the incentives of
	her opponents. First, the expected target of a player is a function of the policies of her opponents. Second, the
	policy of a player determines the correlation between her individual outcome and her opponents’
	outcomes. The follower in a pair of players is the one with the closest policy to the
	status quo. Given a pair of players with different policies, the only player whose
	policy directly affects the covariance is the follower, not the leader. As a result,
	the follower has an extra incentive to explore by choosing a policy in the direction
	of the leader. The new incentive of the follower is the source of conformity. The
	leader-follower relationship induces an asymmetry among players that interacts with
	the network of connections: the follower is pulled away from the status quo by the
	leader, whereas the leader is pushed towards the status quo.
	
	New coordination problems arise in complex environments. The source of equilibrium multiplicity is the presence of endogenous ``kinks'' in payoffs. Intuitively, there is a premium to choosing the same policy as another player because two individual outcomes are the same if the policies are the same. Hence, coordination problems are linked to the leader-follower relationship: by choosing the same policy as an opponent, a player nullifies the asymmetry. The location of the kinks is determined in equilibrium: a player's payoff has a kink at every policy of an opponent. 

	Complexity has implications for the management of organizations with decentralized
	authority, which typically involve practices such as co-op advertising and multi-branding.
	In a multi-division organization, each division manager faces a tradeoff between coordinating other managers and adapting
	to idiosyncratic needs. Moreover, communication frictions create uncertainty in the implementation of managerial instructions. I show that an organization with
	decentralized authority can implement profit maximization in sufficiently complex
	environments (Proposition \ref{prop:org}.) The result points to complexity as a rationale for decentralized organizations.

	\paragraph{Related literature}
	 \citet{callander_searching_2011} introduces the model of complexity to study a dynamic exploration-exploitation tradeoff. The main role of the covariance structure in the dynamic interaction is to discipline learning; see \citet{bardhi_learning_2026} for a survey.  \citet{cetemen_collective_2023} study a similar	complex environment, in which discoveries are correlated over time and members of a team contribute resources for joint exploration. I contribute to the complexity literature by studying coordination motives: the model abstracts from intertemporal incentives and learning, and focuses on coordination among players, each controlling an individual uncertain outcome.\footnote{\citet{garfagnini_uncertainty_2018} studies the rich welfare properties of the Brownian-motion structure for a	network game with a ``degenerate'' covariance: the decision-outcome mappings
	 	are player-specific independent Brownian motions.}

	 Other work considers strategic interactions and Gaussian processes. In particular, the covariance structure has a direct role in \citet{bardhi_local_2023} and \citet{bardhi_selective_2023}; in these models, a principal	incentivizes agents to provide information about the outcome function. These authors study covariances that are characterized by the ``nearest-attribute''	property, including the Brownian covariance. I focus on the Brownian covariance because of two characteristics. First, the Brownian covariance preserves the strategic
	complementarities of the coordination game (Lemma \ref{lem:ID}). Second, this covariance
	contains the leader-follower asymmetry that explains conformity in a simple way (Section \ref{sec:twoplayers}). Other covariances are ``asymmetric'' but not supermodular (e.g.,
	squared-exponential covariance), and vice versa (squared-polynomial).

	The literature on coordination games with quadratic ex-post payoffs includes
	models of oligopolistic competition, peer effects, and network games \citep{jackson_games_2015}. Complexity introduces coordination problems under the standard
	upper bound on the strength of coordination motives. Moreover, complexity makes
	best responses nonlinear, due to kinks in the expected payoffs.
	
%
	
	\section{Model}\label{sec:model}
	We study a coordination game played by $n$ players. 
	
	The utility of player $i\in I\coloneqq \{1, \dots, n\}$ from the profile of outcomes $ \bm x = (x_1, \dots, x_n) \in\mathbb R^ n$ is 
	\begin{align*}
		u_i(\bm x) = -\Big (d_i + \sum _{j\in I}g_{ij} x_j- x_i \Big )^2,
	\end{align*}
	in which $d_i$ is the favorite outcome of player $i$, and $g_{ij}\ge 0$ determines how much the target $d_i + \sum _{j\in I}g_{ij} x_j$ weighs the outcome of opponent $j$, with $g_{ij}= 0$ if $i=j$. We denote by $\bm G$ the matrix with $i$-$j$ element $g_{ij}$, by $\bm I$ the identity matrix, and assume that $\bm I-\bm G$ is invertible.
	
	Every player $i$ chooses a policy $p_i\in P\coloneqq [p_0, \infty)$, for a \emph{status-quo policy $p_0$}. The outcome corresponding to policy $p$ is given by the outcome function $X \colon  \mathbb R \to \mathbb R$, evaluated at $p$. The outcome function is the realized path of a Brownian motion with drift $\mu <0$, variance parameter $\sigma ^2>0 $, and starting point $(p_0, X (p_0))$; see Figure \ref{fig:outcomefunction}.	\begin{figure}[tp]
		\centering
		\subfloat[The outcome function $X$ maps individual policies to individual outcomes and is given by the realized path of a Brownian motion.\protect	\label{fig:outcomefunction}]{
\includegraphics[width=0.4\textwidth]{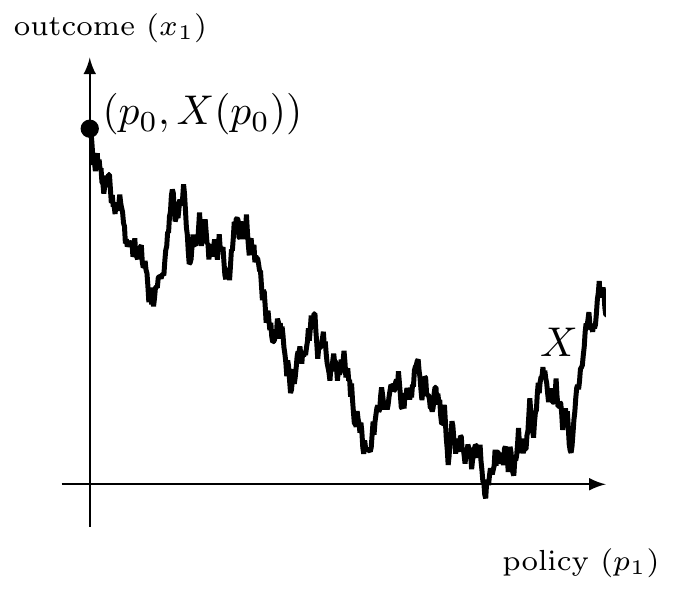}
		}\quad 
		\subfloat[In isolation, player 1 chooses a policy with expected outcome $d_1+k$, which solves her mean-variance tradeoff, or the status-quo policy.\protect\label{fig:supex:B}]{
			\begin{tikzpicture}
				\begin{axis}[
					declare function={
						xnaught=4;
						k=0.6;
						m=-2;
						g=0.3333;
						done=2;
						dtwo=1;
						x1eq1 = (done + k* - xnaught) / m;
						Ex1eq1 = xnaught + m*x1eq1;
					},
					ylabel={ outcome ($x_1$)},
					xlabel={ policy ($p_1$)},
					xtick={x1eq1},
					xticklabels = { $q^1$},
					ytick = {done, Ex1eq1, xnaught},
					yticklabels = { $d_1$,  $d_1+k$,  $X_0$},
					xmax=1.6,
					xmin=0,
					ymax=4.5,
					ymin=1,
					legend pos=north west,
					extra x ticks={0},
					extra x tick labels={ $p_0$},
					]
					\addplot[domain=0:4, black,ultra  thick] {xnaught +m*x};\addlegendentry{$E(X(\cdot))$};
					\addplot[domain=0:x1eq1, blue,very  thick, dashed] {Ex1eq1};
					\draw[blue, very thick, dashed] (axis cs: x1eq1 , 0) -- (axis cs:x1eq1,Ex1eq1); 
				\end{axis}
			\end{tikzpicture}
		}
		\caption{Panel (a) illustrates the outcome function; panel (b) illustrates the equilibrium without coordination motives, that is, with $g_{ij}=0$.}
		\label{fig:supex}
	\end{figure}
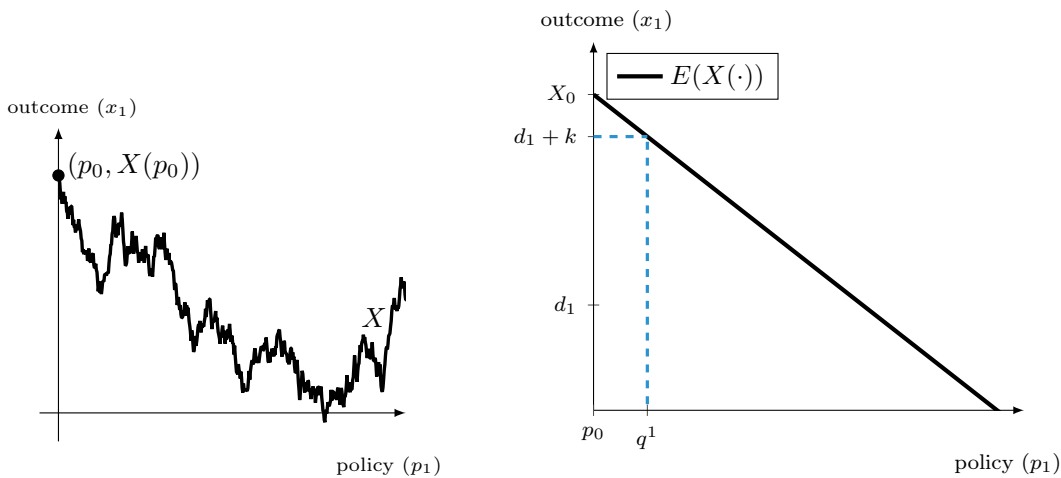 Players do not observe $X$, but they know the underlying stochastic process. Player $i$ evaluates a policy profile $\bm p \in P^n$ by its expected utility $U_i(\bm p) \coloneqq  E(u_i(\bm X(\bm p)))$, letting $\bm X(\bm p)$ denote the outcome profile $(X_1(p_1), \dots, X_n(p_n))$.

	We study the strategic-form game $\Gamma = \langle I, (P, U_i)_{i\in I} \rangle $. The policy profile $\bm p$ is an \emph{equilibrium} if it is a Nash equilibrium of $\Gamma$.

	\paragraph{Analysis}

	The Brownian motion disciplines the uncertainty perceived by players. If a player opts for the status-quo policy, then her outcome is the \emph{status-quo outcome $X(p_0)$} with certainty. If a player opts for a different policy than the status quo, then she knows her outcome only up to a Gaussian noise. The drift parameter $\mu$ determines how the expected outcome $E(X(p))$ changes with the policy $p$.  The variance parameter $\sigma ^2$ determines how the variance of the outcome $X(p)$ increases with the distance of $p$ from the status-quo policy $p_0$. These relationships are linear due to the Brownian motion:
	\begin{align*}
		E(X(p))=X(p_0)+\mu (p-p_0) \ \text{and} \ \var (X(p)) = \sigma^2(p-p_0).
	\end{align*}
	The complexity parameter $k\coloneqq  \sigma ^2 / 2|\mu|$ captures how the uncertainty of an outcome changes as its expected value moves away from the status quo.

	The joint distribution of the outcomes in $\bm X(\bm p)$ is Gaussian. This distribution is determined by the expectation, variance, and covariance of outcomes as functions of the policies. By the independent-increment property of the Brownian motion, the covariance is the variance of the outcome with the closest policy to the status quo,
	\begin{align*}
		\cov (X(p_i), X(p_j))=\sigma ^2(\min\{p_i, p_j\}-p_0).
	\end{align*}


	If player $i$ chooses her policy in isolation, i.e., if $g_{ij}=0$ for all $j$, then she faces a simple mean-variance tradeoff (Figure \ref{fig:supex:B}). The solution is either the status quo or, if uncertainty is not too overwhelming, the policy $p_i$ such that $E(X(p_i)) = d _i + k  $. The optimal policy reflects a status-quo bias. First, for high $k$, player $i$ does not find that decreasing the expected outcome below $X(p_0)$ outweighs the cost $\var (X(p_i))$ for any $p_i>p_0$. Second, conditional on incurring uncertainty with a policy $p_i>p_0$, the mean-variance tradeoff  is smooth, and is solved when the expected outcome of $p_i$ is in between the favorite outcome and the status-quo outcome. Complexity $k $ captures the intensity of status-quo bias for isolated players that find it optimal to incur some uncertainty. Specifically, let $\bm G$ be a matrix of zeroes. There exists a unique equilibrium $\bm q\in P^n$, and $\bm q$ is such that  $E(X(q_i)) = \min \{X(p_0), d_i+k\}$ for all $i$. The one-period restriction of the dynamic model in \citet{callander_searching_2011} has these features.

	The second benchmark for the analysis is the game without complexity. The unique equilibrium is described, e.g., in \citet{ballester_who_2006}. The normal-form game $\Gamma^0 \coloneqq    \langle I; ( u_i(E(\bm X(\cdot))), P)_{i\in I} \rangle $ captures a game without uncertainty because players perceive the outcome function as given by the drift line $p \mapsto X(p_0)+\mu(p-p_0)$. If $\bm I-\bm G$ is invertible, then there exists a unique Nash equilibrium of $\Gamma^0 $. The equilibrium is determined by the weighted centralities in the network of players represented by $\bm G$. Specifically, the unique equilibrium, if interior, is the policy profile $\bm p $ such that $E(\bm X(\bm p))= \bm \beta \coloneqq (\bm I-\bm G)^{-1}\bm d$. Hence, the equilibrium expected outcomes and, a fortiori, the policies are determined by the weighted centralities.

	In general, a given distance between the expected outcome of player $i$ and the status quo is ``inexpensive,'' in terms of uncertainty, if it comes with high covariance between player $i$'s own outcome and the outcomes of other players. For example, the payoff of 1 in the two-player case is
	\begin{align*}
		U_1(\bm p) 
		= - &(d_1+  g_{12}E( X (p_{2}))  -E(X(p_1)))^2 -  \var ( X (p_1))  \\
		&+2 g_{12} \cov (X (p_1),X (p_{2})) -g_{12}^2\var(X (p_2)).
	\end{align*} 
 	Intuitively, the payoff increases with the covariance between outcomes, given that $g_{12}$ is nonnegative.  However, the impact of player $1$ on the covariance depends on the relative position of the policies. Specifically, the policy of player $1$ fully determines the covariance if it is below $p_2$, and is irrelevant if it exceeds $p_2$. For this reason, the mean-variance tradeoff is kinked as a function of $p_1$. The location of the kinks is endogenous: the payoff of player $1$ has a kink at the policy of player $2$.

	The covariance term $\cov (X (p_i),X (p_{j})) $ is  concave as a function of the policies, so an equilibrium exists by standard arguments.
	
	\begin{proposition}\label{prop:existence}
		There exists an equilibrium.
	\end{proposition}

	The covariance of the Brownian motion exhibits increasing differences in the pair of policies, so $U_i$ exhibits increasing differences.

	\begin{lemma}\label{lem:ID}
		The payoff $U_i$ exhibits increasing differences  in $(p_i, \bm p_{-i})$.
	\end{lemma}

	Two observations explain this result. First, the utility from a profile of outcomes $u_i(\bm x)$ satisfies increasing differences in $x_i$ and $x_j$, which implies that the utility from a profile of expected outcomes $u_i(E(\bm X(\bm p)))$ satisfies increasing differences in $p_i$ and $p_j$. Second, the payoff $U_i(\bm p)$ adds variance and covariance elements to $u_i(E(\bm X(\bm p)))$, so $U_i(\bm p)$ exhibits increasing differences in $p_i$ and $p_j$.
	
	In Section \ref{sec:twoplayers}, we analyze a two-player example. The analysis uses the increasing-differences property to describe the equilibria the game.\footnote{Any Gaussian process with a covariance function that satisfies increasing differences preserves the sign of the strategic externality of the game; see the Supplemental appendix.} In Section \ref{sec:equilibrium}, I characterize the equilibria using the intuitions from Section \ref{sec:twoplayers}.

	\section{Two players}\label{sec:twoplayers}
	
	This section considers the case with $n=2$. This case highlights the role of the covariance structure. For simplicity, we assume that $g\coloneqq g_{12} = g_{21}$, and that $X(p_0)-k > d_1 > d_2 + 2gk$. The latter restriction guarantees that player 1 is closer to the status quo than player 2 in equilibrium, because: the favorite outcome $d_1$ is sufficiently closer to the status quo than $d_2$, and equilibrium policies are interior.

	In the game $\Gamma^0 $ without uncertainty, the best response of player $i$ to $p_j$ is the policy $p_i$ such that
	\begin{align}
		E(X(p_i)) = d_i + g E(X(p_j)).
		\label{eq:2playercert}
	\end{align}
	\begin{figure}[tp]
		\centering
		\subfloat[The equilibrium in $\Gamma ^0$ is the intersection point of the dotted best responses ($\operatorname {BR}_i^0$); the equilibrium in $\Gamma ^1$ is the intersection of the solid best responses ($\operatorname {BR}_i^1$).\protect\label{fig:twoplayer_nouncertain}]{
			\begin{tikzpicture}
				\begin{axis}[
					declare function={
						xnaught=4;
						k=0.6;
						m=-2;
						g=0.3333;
						done=2;
						dtwo=1;
						BRnaughtone(\x)  = done/m - (xnaught*(1-g))/m + g*\x;
						BRoneone(\x)     = done/m - (xnaught*(1-g))/m + k/m + g*\x;
						BRone(\x)  = done/m - (xnaught*(1-g))/m + k*(1-2*g)/m + g*\x;
						BRoneoneIntersect45 =	( done + k - xnaught*(1-g) ) / ( m*(1-g) );
						BRoneIntersect45 =	( done + k*(1 - 2*g) - xnaught*(1-g) ) / ( m*(1-g) );
						InverseBRnaughttwo(\x)  = (\x - (dtwo/m - (xnaught*(1-g))/m)) / g;
						InverseBRonetwo(\x)     = (\x - (dtwo/m - (xnaught*(1-g))/m + k/m)) / g;
						InversealtBRonetwo(\x)  = (\x - (dtwo/m - (xnaught*(1-g))/m + k*(1-2*g)/m)) / g; 
						InverseBRonetwoIntersect45 =	( dtwo + k - xnaught*(1-g) ) / ( m*(1-g) );
						InversealtBRonetwoIntersect45 =	( dtwo + k*(1-2*g) - xnaught*(1-g) ) / ( m*(1-g) );
						x1eq0 = (done + g*dtwo - xnaught*(1 - g*g)) / (m*(1 - g*g));
						x2eq0 = (dtwo + g*done - xnaught*(1 - g*g)) / (m*(1 - g*g));
						x1eq1 = (done + g*dtwo + k*(1+g) - xnaught*(1 - g*g)) / (m*(1 - g*g));
						x2eq1 = ( dtwo + g*done + k*(1+g) - xnaught*(1 - g*g) ) / ( m*(1 - g*g) );
						x1eq = (done + g*dtwo + k*(1-g) - xnaught*(1 - g*g)) / (m*(1 - g*g));
						x2eq = ( dtwo + g*done + k*(1 + g - 2*g*g) - xnaught*(1 - g*g) ) / ( m*(1 - g*g) );
					},
					ylabel={ policy ($p_1$)},
					xlabel={ policy ($p_2$)},
					xtick={x2eq1, x2eq0},
					xticklabels = { $p_2^1$,  $p_2^0$},
					ytick = {x1eq1, x1eq0},
					yticklabels = { $p_1^1$, $p_1^0$},
					xmax=1.3,
					xmin=0,
					ymax=1.3,
					ymin=0,
					legend pos=north west,
					]
					\addplot[domain=0:4, blue,  very thick, densely dotted] {max(0.01, BRnaughtone(\x))};\addlegendentry{$\operatorname{BR}^0_1$};
					\addplot[domain=0:4, red,  very thick, densely dotted]{InverseBRnaughttwo(\x)}; \addlegendentry{$(\operatorname{BR}_2^0)^{-1}$};
					\addplot[domain=0:4, blue,ultra  thick] {max(0.01, BRoneone(\x))};\addlegendentry{$\operatorname{BR}_1^1$};
					\addplot[domain=0:4, red,ultra  thick] {InverseBRonetwo(\x)};\addlegendentry{$(\operatorname{BR}_2^1)^{-1}$};
					\addplot[domain=0:x2eq0, blue, thick, dotted] {x1eq0};
					\draw[red, thick, dotted] (axis cs: x2eq0, 0) -- (axis cs: x2eq0, x1eq0); 
					\addplot[domain=0:x2eq1, blue, thick, dashed] {x1eq1};
					\draw[red, thick, dashed ] (axis cs: x2eq1, 0) -- (axis cs: x2eq1, x1eq1); 
				\end{axis}
			\end{tikzpicture}
		} \hspace{0.25em} 
		\subfloat[In the equilibrium $\bm p^1$ of $\Gamma ^1$, players are closer to the status quo that in the equilibrium $\bm p^0$ of $\Gamma ^0$. The introduction of pure noise shifts best responses by the same amount, so the distance between players is the same.\protect\label{fig:twoplayer_nouncertain_eq}]{
			\begin{tikzpicture}[scale=0.9]
				\begin{axis}[
					declare function={
						xnaught=4;
						k=0.75;
						m=-2;
						g=0.3333;
						done=2;
						dtwo=1;
						BRnaughtone(\x)  = done/m - (xnaught*(1-g))/m + g*\x;
						BRoneone(\x)     = done/m - (xnaught*(1-g))/m + k/m + g*\x;
						BRone(\x)  = done/m - (xnaught*(1-g))/m + k*(1-2*g)/m + g*\x; 
						InverseBRnaughttwo(\x)  = (\x - (dtwo/m - (xnaught*(1-g))/m)) / g;
						InverseBRonetwo(\x)     = (\x - (dtwo/m - (xnaught*(1-g))/m + k/m)) / g;
						InversealtBRonetwo(\x)  = (\x - (dtwo/m - (xnaught*(1-g))/m + k*(1-2*g)/m)) / g; 
						x1eq0 = (done + g*dtwo - xnaught*(1 - g*g)) / (m*(1 - g*g));
						x2eq0 = (dtwo + g*done - xnaught*(1 - g*g)) / (m*(1 - g*g));
						x1eq1 = (done + g*dtwo + k*(1+g) - xnaught*(1 - g*g)) / (m*(1 - g*g));
						x2eq1 = ( dtwo + g*done + k*(1+g) - xnaught*(1 - g*g) ) / ( m*(1 - g*g) );
						x1eq = (done + g*dtwo + k*(1-g) - xnaught*(1 - g*g)) / (m*(1 - g*g));
						x2eq = ( dtwo + g*done + k*(1 + g - 2*g*g) - xnaught*(1 - g*g) ) / ( m*(1 - g*g) );
						Ex1eq0 = xnaught + m*x1eq0;
						Ex2eq0 = xnaught + m*x2eq0;
						Ex1eq1 = xnaught + m*x1eq1;
						Ex2eq1 = xnaught + m*x2eq1;
						Ex1eq = xnaught + m*x1eq;
						Ex2eq = xnaught + m*x2eq;
					},
					ylabel={ outcome ($x_i$)},
					xlabel={ policy ($p_i$)},
					xtick={x2eq1, x2eq0, x1eq1, x1eq0},
					xticklabels = { $p_2^1$,  $p_2^0$,  $p_1^1$,  $p_1^0$},
					ytick = {Ex1eq1, Ex1eq0, Ex2eq1, Ex2eq0},
					yticklabels = { $E(X(p_1^1))$,  $E(X(p_1^0))$,  $E(X(p_2^1))$,  $E(X(p_2^0))$},
					xmax=1.3,
					xmin=0,
					ymax=4.3,
					ymin=1.8,
					legend pos=north west,
					]
					\addplot[domain=0:4, black,ultra  thick] {xnaught +m*x}; \addlegendentry{$E(X(\cdot))$};
					\addplot[domain=0:x1eq0, blue, very thick,  dotted] {Ex1eq0};
					\draw[blue, very thick,   dotted] (axis cs: x1eq0 , 0) -- (axis cs:x1eq0,Ex1eq0); 
					\addplot[domain=0:x2eq0, red,  very thick,   dotted] {Ex2eq0};
					\draw[red,  very thick,   dotted] (axis cs: x2eq0 , 0) -- (axis cs:x2eq0,Ex2eq0); 
					\addplot[domain=0:x1eq1, blue,  very thick, dashed] {Ex1eq1};
					\draw[blue,  very thick, dashed] (axis cs: x1eq1, 0) -- (axis cs:x1eq1,Ex1eq1); 
					\addplot[domain=0:x2eq1, red,  very thick, dashed] {Ex2eq1};
					\draw[red,  very thick, dashed] (axis cs: x2eq1, 0) -- (axis cs:x2eq1,Ex2eq1); 
				\end{axis}
			\end{tikzpicture}
		}
		\caption{Panel (a) illustrates the best responses in the game without uncertainty $\Gamma ^0$ and in the game with pure noise $\Gamma ^1$, panel (b) illustrates the equilibria.}\protect\label{fig:twoplayer_nouncertainboth}
	\end{figure}
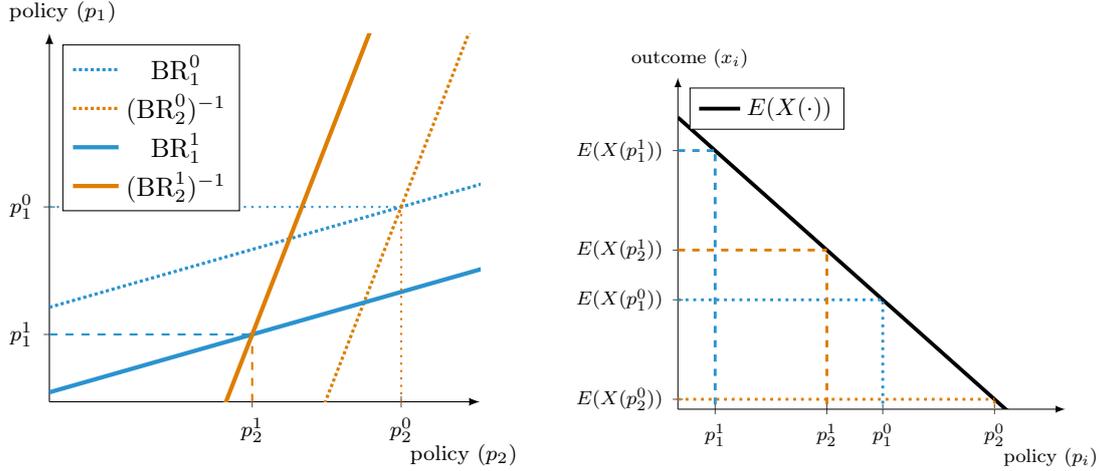
	The equilibrium $\bm p^0$ is determined by the centralities as the unique policy profile $\bm p$ such that $E(\bm X(\bm p)) = (\bm I - \bm G)^{-1}\bm d$. The equilibrium distance between players is captured by $D\coloneqq E(X(p_1^0))-E(X(p_2^0))$.
	
	For the sake of intuition, we consider an intermediate scenario between the no-uncertainty case in $\Gamma ^0$ and the game $\Gamma$. We consider the game $\Gamma ^1$ with ``pure noise:'' the outcome $X_1(p_1)$ is independent of the outcome $X_2(p_2)$ for every policy $p_1$ of player 1 and $p_2$ of player 2. Additionally, the variance and mean of outcomes are the same as in $\Gamma$. Formally, this structure of uncertainty arises from player-specific independent Brownian motions, and is considered in the Supplemental appendix. The best response of player $i$ to $p_j$ is the policy $p_i$ such that
	\begin{align}
		E(X_i(p_i)) = d_i + gE(X_j(p_j)) +k. \label{eq:2playernoise}
	\end{align}
	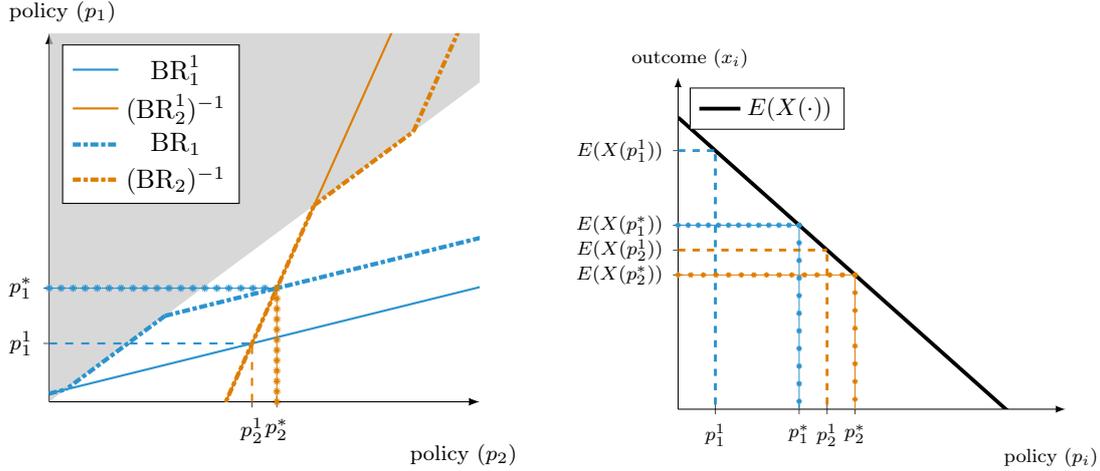
\begin{figure}[tp]
		\centering
		\subfloat[The equilibrium in $\Gamma ^1$ is the intersection point of the solid best responses, $\operatorname {BR}_i^1$, the equilibrium in the principal game $\Gamma$ if the intersection of the dash-dotted best responses, $\operatorname {BR}_i$. In the gray region, best responses take into account that player 1 is the leader.\protect\label{fig:twoplayer}]{
			\begin{tikzpicture}
				\begin{axis}[
					ylabel={ policy ($p_1$)},
					xlabel={ policy ($p_2$)},
					declare function={
						xnaught=4;
						k=0.6;
						m=-2;
						g=0.3333;
						done=2;
						dtwo=1;
						BRnaughtone(\x)  = done/m - (xnaught*(1-g))/m + g*\x;
						BRoneone(\x)     = done/m - (xnaught*(1-g))/m + k/m + g*\x;
						BRone(\x)  = done/m - (xnaught*(1-g))/m + k*(1-2*g)/m + g*\x; 
						BRoneoneIntersect45 =	( done + k - xnaught*(1-g) ) / ( m*(1-g) );
						BRoneIntersect45 =	( done + k*(1 - 2*g) - xnaught*(1-g) ) / ( m*(1-g) );
						InverseBRnaughttwo(\x)  = (\x - (dtwo/m - (xnaught*(1-g))/m)) / g;
						InverseBRonetwo(\x)     = (\x - (dtwo/m - (xnaught*(1-g))/m + k/m)) / g;
						InversealtBRonetwo(\x)  = (\x - (dtwo/m - (xnaught*(1-g))/m + k*(1-2*g)/m)) / g; 
						InverseBRonetwoIntersect45 =	( dtwo + k - xnaught*(1-g) ) / ( m*(1-g) );
						InversealtBRonetwoIntersect45 =	( dtwo + k*(1-2*g) - xnaught*(1-g) ) / ( m*(1-g) );
						x1eq0 = (done + g*dtwo - xnaught*(1 - g*g)) / (m*(1 - g*g));
						x2eq0 = (dtwo + g*done - xnaught*(1 - g*g)) / (m*(1 - g*g));
						x1eq1 = (done + g*dtwo + k*(1+g) - xnaught*(1 - g*g)) / (m*(1 - g*g));
						x2eq1 = ( dtwo + g*done + k*(1+g) - xnaught*(1 - g*g) ) / ( m*(1 - g*g) );
						x1eq = (done + g*dtwo + k*(1-g) - xnaught*(1 - g*g)) / (m*(1 - g*g));
						x2eq = ( dtwo + g*done + k*(1 + g - 2*g*g) - xnaught*(1 - g*g) ) / ( m*(1 - g*g) );
					},
					xtick={x2eq1, x2eq},
					xticklabels = { $p_2^1$,  $p_2^*$},
					ytick = {x1eq1, x1eq},
					yticklabels = { $p_1^1$,  $p_1^*$},
					xmax=1.3,
					xmin=0,
					ymax=1.5,
					ymin=0,
					]
					\addplot[domain=0:4, blue,   thick, solid] {max(0.01, BRoneone(\x))};\addlegendentry{ $\operatorname{BR}_1^1$};
					\addplot[domain=0:4, red,   thick, solid] {InverseBRonetwo(\x)};\addlegendentry{ $(\operatorname{BR}_2^1)^{-1}$};
					\addplot[domain=BRoneIntersect45:4, blue, ultra  thick, densely dashdotted] {max(0.01, BRone(\x))};\addlegendentry{$\operatorname{BR}_1$};
					\addplot[domain=0:InverseBRonetwoIntersect45, red,  ultra thick, densely dashdotted] {InverseBRonetwo(\x)};\addlegendentry{$(\operatorname{BR}_2)^{-1}$};
					\addplot[domain=0:x2eq1, blue, thick, dashed] {x1eq1};
					\draw[red,thick, dashed] (axis cs: x2eq1, 0) -- (axis cs: x2eq1, x1eq1); 
					\addplot[domain=0:x2eq, blue,  thin, mark=10-pointed star,  mark size=1.4pt, mark options={fill=blue, line width=0.26pt}, samples=20] {x1eq};
					\addplot[red,  thin, mark=10-pointed star, mark size=1.4pt, mark options={fill=red, line width=0.26pt}, samples=12] plot[domain=0:1] (x2eq, {x1eq * x});
					\draw[blue, ultra thick, densely dashdotted] (axis cs: BRoneIntersect45, BRoneIntersect45) -- (axis cs: BRoneoneIntersect45, BRoneoneIntersect45); 
					\addplot[domain=0:BRoneoneIntersect45, blue, ultra  thick, densely dashdotted] {max(0.01, BRoneone(\x))};
					\draw[red,  ultra thick, densely dashdotted] (axis cs: InverseBRonetwoIntersect45, InverseBRonetwoIntersect45) -- (axis cs: InversealtBRonetwoIntersect45, InversealtBRonetwoIntersect45); 
					\addplot[domain=InversealtBRonetwoIntersect45:4, red,  ultra thick, densely dashdotted] {InversealtBRonetwo(\x)};
					\addplot[domain=0:4,name path=diag, white, opacity=0] {x};
					\addplot[domain=0:4, name path=upper, opacity = 0] {1.5}; 
					\addplot[gray!30, opacity=0.7] fill between[of=diag and upper];
				\end{axis}
			\end{tikzpicture}
		}\hspace{0.25em}
		\subfloat[In the equilibrium $\bm p^*$ of $\Gamma $, players are more distant from the status quo that in the equilibrium $\bm p^1$ of $\Gamma ^1$. Correlation shifts only the best response of the follower; this asymmetry implies that the distance between the players is different in the two case.\protect\label{fig:twoplayer_eq}]{
			\begin{tikzpicture}[scale=0.9]
				\begin{axis}[
					declare function={
						xnaught=4;
						k=0.75;
						m=-2;
						g=0.3333;
						done=2;
						dtwo=1;
						BRnaughtone(\x)  = done/m - (xnaught*(1-g))/m + g*\x;
						BRoneone(\x)     = done/m - (xnaught*(1-g))/m + k/m + g*\x;
						BRone(\x)  = done/m - (xnaught*(1-g))/m + k*(1-2*g)/m + g*\x; 
						InverseBRnaughttwo(\x)  = (\x - (dtwo/m - (xnaught*(1-g))/m)) / g;
						InverseBRonetwo(\x)     = (\x - (dtwo/m - (xnaught*(1-g))/m + k/m)) / g;
						InversealtBRonetwo(\x)  = (\x - (dtwo/m - (xnaught*(1-g))/m + k*(1-2*g)/m)) / g; 
						x1eq0 = (done + g*dtwo - xnaught*(1 - g*g)) / (m*(1 - g*g));
						x2eq0 = (dtwo + g*done - xnaught*(1 - g*g)) / (m*(1 - g*g));
						x1eq1 = (done + g*dtwo + k*(1+g) - xnaught*(1 - g*g)) / (m*(1 - g*g));
						x2eq1 = ( dtwo + g*done + k*(1+g) - xnaught*(1 - g*g) ) / ( m*(1 - g*g) );
						x1eq = (done + g*dtwo + k*(1-g) - xnaught*(1 - g*g)) / (m*(1 - g*g));
						x2eq = ( dtwo + g*done + k*(1 + g - 2*g*g) - xnaught*(1 - g*g) ) / ( m*(1 - g*g) );
						Ex1eq0 = xnaught + m*x1eq0;
						Ex2eq0 = xnaught + m*x2eq0;
						Ex1eq1 = xnaught + m*x1eq1;
						Ex2eq1 = xnaught + m*x2eq1;
						Ex1eq = xnaught + m*x1eq;
						Ex2eq = xnaught + m*x2eq;
					},
					ylabel={ outcome ($x_i$)},
					xlabel={ policy ($p_i$)},
					xtick={x2eq1, x2eq, x1eq1, x1eq},
					xticklabels = { $p_2^1$,  $p_2^*$,  $p_1^1$,  $p_1^*$},
					ytick = {Ex1eq1, Ex1eq, Ex2eq1, Ex2eq},
					yticklabels = { $E(X(p_1^1))$,  $E(X(p_1^*))$, $E(X(p_2^1))$,  $E(X(p_2^*))$},
					xmax=1.3,
					xmin=0,
					ymax=4.3,
					ymin=1.8,
					legend pos=north west,
					]
					\addplot[domain=0:4, black,ultra  thick] {xnaught +m*x};\addlegendentry{$E(X(\cdot))$};
					\addplot[domain=0:x1eq1, blue, very thick, dashed] {Ex1eq1};
					\draw[blue,  very thick, dashed] (axis cs: x1eq1, 0) -- (axis cs:x1eq1,Ex1eq1); 
					\addplot[domain=0:x2eq1, red, very thick, dashed] {Ex2eq1};
					\draw[red, very thick, dashed] (axis cs: x2eq1, 0) -- (axis cs:x2eq1,Ex2eq1); 
					\addplot[domain=0:x1eq, blue,  very thin, mark=10-pointed star, mark size=1.2pt, mark options={fill=blue, line width=0.25pt}, samples=13] {Ex1eq};
					\addplot[blue, very thin, mark=10-pointed star, mark size=1.2pt, mark options={fill=red, line width=0.25pt}, samples=25]  plot[domain=0:1] (x1eq, {Ex1eq * x});
					\addplot[domain=0:x2eq, red, very thin, mark=10-pointed star, mark size=1.2pt, mark options={fill=blue, line width=0.25pt}, samples=15] {Ex2eq};
					\addplot[red, very thin, mark=10-pointed star, mark size=1.2pt, mark options={fill=red, line width=0.25pt}, samples=22]  plot[domain=0:1] (x2eq, {Ex2eq * x});
				\end{axis}
			\end{tikzpicture}
		}
		\caption{Panel (a) illustrates the best responses in the game with pure noise $\Gamma ^1$ and in the principal game $\Gamma$ --- the border between the gray and white regions is the main diagonal; panel (b) illustrates the equilibria.}.\protect\label{fig:twoplayer_both}
	\end{figure}
	Best responses shift inwards, with respect to the noiseless case, because of the mean-variance tradeoff introduced by the variance of outcomes. The game $\Gamma ^1$ is of strategic complements, because the utility function $u_i$ exhibits increasing differences in the own outcome $x_i$ and the outcome of an opponent $x_j$. Specifically, the game without uncertainty $\Gamma^0$ is a game of strategic complements, and the game $\Gamma ^1$ simply adds a cost of moving away from $p_0$ to each payoff function, without generating externalities. Therefore, the equilibrium involves lower policies than the case of no uncertainty. Because the best responses are linear, the equilibrium $\bm p^1$ is determined as the unique policy profile $\bm p$ such that $E(\bm X(\bm p))  = (\bm I - \bm G)^{-1}(\bm d+k\bm 1)$. By direct computation, the distance between players is unchanged: $D = E(X(p_1^1))-E(X(p_2^1)) $.

	The equilibrium in  $\Gamma^1$  features a magnified status-quo bias. The interaction of noise and coordination motives makes players closer to the status quo than in $\Gamma ^0$. The status-quo bias is magnified because:  the equilibrium expected outcome $E(X(p_i^1))$ is obtained by adding the positive amount $\frac{1}{1-g}k$ to the no-uncertainty outcome $E(X(p_i^0))$, that is
	\begin{align*}
	E(X(p_i^1))	 = 	E(X(p_i^0)) + \frac{1}{1-g}k.
	\end{align*}
	In the new term $\frac{1}{1-g}k$, the social multiplier $\frac{1}{1-g}$ \citep{jackson_games_2015} multiplies the status-quo bias $k$ \citep{callander_searching_2011}.

	In the game $\Gamma$, the outcomes are correlated across players due to the common outcome function $X$. By the independent-increment property, the covariance between outcomes is given by $\cov (X(p_1), X(p_2)) = \min_i  \var (X(p_i))$. So player 2 has no control over the covariance locally to the policy profile $\bm p^1$. Hence,  the  best response of player 2 to $p_1^1$ is unchanged from the case of pure noise (equation \ref{eq:2playernoise}). In contrast, player 1 has an extra exploration motive due to her control over the covariance, so the best response to $p_2^1$ is the policy $p_1$ such that
		\begin{align}
		E(X_1(p_1)) = d_1 + gE(X_2(p_2^1)) +k - 2gk.
			\label{eq:2playergame}
	\end{align}
	The best response of player 1 shifts upwards by the introduction of the covariance in outcomes across players: the new term  $- 2gk$ reflects an exploration motive. The best-response shift depends on coordination motives and complexity. For fixed $k$, the larger $g$, the more important the exploration motive relative to the status-quo bias. Hence, the exploration motive of player 1 interacts with the strategic complementarity of the game. In equilibrium, both players are farther away from the status quo than with independent noise, and the equilibrium $\bm p^*$ is given by  the unique policy profile $\bm p$ such that
	\begin{align*}
		E(\bm X(\bm p)) = (\bm I - \bm G)^{-1} \Bigg ( \bm d + \bm 1 k - \begin{pmatrix}
			2gk\\0
		\end{pmatrix} \Bigg  )
	\end{align*}
	Hence, the exploration motive of player 1 acts against the magnified status-quo bias that is due to pure noise. In fact, we have
	\begin{align*}
		E(X(p_i^*) ) = \begin{cases}
			E(X(p_i^1))	 -2\frac{g^2}{1-g^2}k\Big (1+[i=1]\frac{1-g}{g}\Big ),
		\end{cases}
	\end{align*}
	so both players are farther away from the status quo than with noise.\footnote{The Iverson bracket $[i=1]$ is 1 if $i=1$ and $0$ otherwise.}
	 Because correlation directly affects the follower's incentives, the follower's exploration motive is stronger than the leader's one.

	Importantly, using the formula in \ref{eq:2playergame} for the best response of player 1, and \ref{eq:2playernoise} for player $2$, presumes that player 1 is closer to the status quo than player 2 in equilibrium.  In what follows, if the policy of player $i$ is higher than the policy of player $j$, we refer to $i$ as the  \emph{leader} and to player $j$ as the \emph{follower}. In general, the leader-follower structure is determined in equilibrium, so best responses are not globally linear.
	
	In equilibrium, the two players are closer together than in the previous cases. In particular, the equilibrium distance between players is $E(X(p_1^*))-E(X(p_2^*))=D-2\frac{g}{1+g}k$, so that
	\begin{align*}
		E(X( p^*_1)) -E( X(p^*_2)) <	E(X( p^1_1)) -E( X(p^1_2)) =	E(X( p^0_1)) -E( X(p^0_2)).
	\end{align*}
	Conformity increases in the complexity of the environment. The intuition for the comparative statics follows from the leader-follower relationship: matching the outcome of a leader becomes more ``cost-effective'' for a follower, relative to targeting a favorite outcome, as complexity increases. This comparative statics is consistent with findings in social psychology. Since \citet{asch_effects_1951}, psychologists observe that conformity ``is far greater on difficult items than on easy ones.'' The ``difficulty'' is typically obtained by asking experimental subjects about their ``certainty of judgement'' \citep{krech_individual_1962}.

	\section{Equilibrium characterization}\label{sec:equilibrium}
	In this section, we characterize the equilibria of the game. 
	
	We define a matrix that keeps track of which player is a follower in every pair of players. Given a policy profile $\bm p$, a matrix $\bm F \in [0,1]^{n\times n}$ is a \emph{$\bm p$-followership} if: $p_i<p_j$ implies $f_{ij}=1$ and $f_{ji}=0$. The definition does not impose requirement on a followership if $p_i=p_j$, and the multiplicity of values of $f_{ij}$ consistent with the definition allows to characterize the equilibrium set. (The symbol $\circ$ denotes entry-wise multiplication.)

	\begin{proposition} \label{prop:equilibrium}
		The policy profile $\bm p$ is an equilibrium if and only if there exists a $\bm p$-followership $\bm F$ such that:
	\begin{align*}
		E(\bm X(\bm p))\le \bm M(\bm d +k(\bm I-2\bm G \circ \bm F)\bm 1),
	\end{align*}
	and, for all $i\in I$, $p_i>p_0$ only if $E( X_i( p_i))=\sum_{j\in I}m_{ij}(d_j +k-2k\sum_{\ell\in I} g_{j\ell }f_{j\ell})$.
\end{proposition}

	For intuition, consider an interior equilibrium $\bm p$. Without any correlation among the outcomes of players, the formula in the proposition reduces to $ E(\bm X(\bm p))= \bm M(\bm d +k\bm 1)$, in which the coordination motives interact with the status-quo bias that arises from the variance in individual outcomes. This formula characterizes the unique equilibrium in terms of the parameters of the game, and extends the analysis of ``pure noise'' in Section \ref{sec:twoplayers}.

	The covariance structure introduces a new term, captured by the vector $\bm v \coloneqq (\bm G \circ \bm F) \bm 1$. This vector collects ``average'' followership roles of players in the policy profile $\bm p$. In particular, $v_i=\sum_{\ell\in I} g_{i\ell }f_{i\ell}$, in which $f_{i\ell}$ records whether $i$ is a follower in the pair of players $(i, \ell )$, and $g_{i\ell }$ scales the followership by how much $i$ cares about coordinating with $\ell $. Importantly, $\bm v$ is constructed from a $\bm p$-followership, so it is not a parameter of the model: the relative position of the policies is determined in equilibrium, for all player pairs. Put differently, the equilibrium equation $ E(\bm X(\bm p))= \bm M(\bm d + k(\bm 1 -2\bm v))$ is a fixed-point condition.

	In a uniform network, a natural guess for the followership supporting an equilibrium is given by making $i$ a follower of $j$ if and only if favorite outcomes are ranked this way. This condition holds for every equilibrium with this network, which generalizes the observations in Section \ref{sec:twoplayers}.

	\begin{corollary}\label{cor:symmetriceq}
		Let $g\ge 0$ be such that  $g_{ij}=g$ for all $i$ and $j\ne i$, and  $\bm p\in(0, \infty)^n$ be an equilibrium. For all players $i$ and $j$,  $p_i<p_j$ implies
		 \begin{align*}
			E (X_i(p_i) - X_j(p_j)) < \beta_i - \beta_j.
		\end{align*}
		Moreover, if $p_1<p_2<\dots$, then
		\begin{align*}
				E (X_i(p_i) - X_{i+1}(p_{i+1}))  =  \beta_i - \beta_{i+1} -2 \frac{g}{1+g }k.
		\end{align*}
	\end{corollary}
	

	Symmetry among players brings about coordination problems in the sense of equilibrium multiplicity. Moreover, the equilibrium set expands as complexity increases if coordination motives are sufficiently strong.
	
	\begin{corollary}\label{cor:multiplicity}
		Let $d_i=0$ for all $i\in I$ and players have the same unweighted centralities given by $\bm u\coloneqq (\bm I-\bm G)^{-1}\bm 1$. Then, $\bm p\in(0, \infty)^n$ is an equilibrium if and only if:
		\begin{align*}
			E( \bm X (\bm p) ) \in [(2\bm 1-\bm u) k, \bm u k].
		\end{align*}
		Moreover, $\bm uk$ is increasing $k$, and $(2- u_i) k$ is increasing in $k$ if and only if $u_i<2$.
	\end{corollary}

	The game admits a potential function \citep{monderer_potential_1996}. Specifically, consider the function $V\colon P^n\to \mathbb R$,  given by
	\begin{align*}
		V(\bm p) = E\Big (\sum _{i\in I} d_iX(p_i) - \frac{1}{2}(X(p_i))^2 +\sum_{j\in I}\frac{1}{2}g_{ij}X(p_{i})X(p_j)\Big ),
 	\end{align*}
	and say that a $\bm p$ followership is \emph{skew-complementary} if $f_{ij}+f_{ji}	=1$ for all $(i, j)\in I^2$. The following result establishes existence and uniqueness of the potential maximizer.\footnote{Uniqueness of the potential-maximizer equilibrium obtains jointly with the multiplicity of equilibria because the potential is not smooth. Specific nondifferentiable potentials are studied as counterexamples to the results for smooth potentials \citep{radner_team_1962, neyman_correlated_1997}. These articles show strictly concave nondifferentiable potentials that are consistent with multiple equilibria. Instead, in games with strictly concave and smooth potentials, the solution to a first-order condition is the unique equilibrium of the game.}

		\begin{proposition}\label{prop:potentialmax}
		If $\bm G$ is symmetric, then $V$ is a potential for $\Gamma$, and there exists a unique potential maximizer. Moreover, the policy profile $\bm p$ maximizes  the potential if and only if there exists a skew-complementary $\bm p$-followership $\bm F$ such that:
		\begin{align*}
			E(\bm X(\bm p))\le \bm M(\bm d +k(\bm I-2\bm G \circ \bm F)\bm 1),
		\end{align*}
		and, for all $i\in I$, $p_i>p_0$ only if $E( X_i( p_i))=\sum_{j\in I}m_{ij}(d_j +k-2k\sum_{\ell\in I} g_{j\ell }f_{j\ell})$.
		\end{proposition}
	
		The characterization of the potential maximizer is important because the utilitarian-maximal profile maximizes the potential of a modified game with twice the degree of coordination motives.\footnote{This observation is known without complexity \citep{jackson_games_2015}.}

	\section{Centralization in organizations}\label{sec:organization}
	This section considers a stylized model of an organization with multiple divisions. The model and notation follow closely \citet{alonso2015}.
	
	The firm consists of two divisions, each producing a different good.  The inverse demand for product $i$ is $q_i\mapsto a_i-\frac{1}{b}q_i$, in which $b>0$ measures the price elasticity of demand. The cost of division $i$ is $c_iq_i-gq_1q_2$, in which: $q_i$ is the quantity produced by division $i$, $g> 0$ measures the degree of cost externalities, and $c_i>0$.  An increase in the quantity produced by one division reduces the marginal cost of the other division. This cost externality, together with the complexity introduced in the following paragraph, generates strategic complementarity as per Lemma \ref{lem:ID}.  The profit of division $i$ given the quantity profile $\bm q$ is
	\begin{align*}
		\pi_i(\bm q) \coloneqq \Big   (a_i-\frac{1}{b}q_i - c_i +gq_{j}\Big )q_i.
	\end{align*}

	Each division chooses a policy $p_i\in P$. The function $X$ maps production policies into quantities. The profits of division $i$ given the pair of policies $\bm p$ are $\Pi_i(\bm p)\coloneqq \pi_i (\bm X (\bm p))$. Divisions set production policies simultaneously and independently in the \emph{production game} $\langle  \{1, 2\}, (P, E( \pi _i (\bm X (\cdot ))))_{i\in \{1, 2\}} \rangle $. The maximization of total profit, $W\coloneqq \Pi_1+\Pi_2$, is an important benchmark in organizational economics  \citep{gibbons_decisions_2013}.  I assume that  $g<1/b$, so that the Hessian of $W$ is negative definite.\footnote{Under this condition, the equilibrium set does not expand as $\sigma ^2$ increases, so Proposition \ref{prop:org} is not due to Corollary \ref{cor:multiplicity}.}

	The following result investigates the compatibility of managerial incentives with total-profit maximization, under the hyopothesis that the status quo is sufficiently bad for the twe division managers.
	\begin{proposition}\label{prop:org}
		There exists a unique policy profile $\bm p^O$ that maximizes $W$.  Moreover, if $X_0\ge  \frac{b a_1 - b c_1 + b g (b a_2 - b c_2) + 2 k}{2 (1 - b g)}$, then there exists $C>0$ such that: $\bm p^O$ is an equilibrium of the production game if and only if $k\ge C$.
	\end{proposition}
	
 	First, observe that total profit corresponds to the utilitarian welfare in the production game, and the production game is best-response equivalent to $\Gamma$ if $2d_i=ba_i-bc_i$ and $bg=2g_{ij}$. Hence, I identify the profit-maximal policy profile by applying Proposition \ref{prop:potentialmax} to the modified production game in which cost externalities are doubled.

 	Three observations follow from the two-player analysis in Section \ref{sec:twoplayers} and key for the result. First, there are multiple equilibria only if players choose the same policy (Figure \ref{fig:twoplayer}.) Moreover, because the profit maximizer is obtained by doubling the coordination motives, profit is not maximized by an equilibrium if there exists a unique equilibrium. Therefore, a necessary condition for profit to be maximized by an equilibrium is equilibrium multiplicity.

 	The second observation is illustrated in Figure \ref{fig:org_policies}. Sufficiently high coordination motives induce equilibrium multiplicity, due to the leader-follower asymmetry (Section \ref{sec:twoplayers}.) Moreover, the cutoff value of $k$ above which multiple equilibria arise is decreasing in $g$: coordination motives magnify the leader-followership asymmetry. So, if there are multiple equilibria, then automatically players choose the same policy under profit maximization.

 	Finally, consider an environment with multiple equilibria, and refer to Figure \ref{fig:org_policies}. There exists a largest and smallest equilibrium by  Lemma \ref{lem:ID} \citep{milgrom_rationalizability_1990}. Intuitively, ``in'' the largest equilibrium, players are exploring more and so are more sensitive to complexity than in other equilibria. Hence, the largest equilibrium, as a single policy, decreases in $k$ faster than the potential-maximizer equilibrium, which in turn decreases faster than the smallest equilibrium. Crucially, the derivative of the potential maximizer with respect to $\sigma ^2$  does not depend on  $g$. In fact, the potential maximizer in a region of multiple equilibria is found by treating the two players as a single agent who solves a mean-variance tradeoff. We conclude that the profit maximizer, as a single policy, is decreasing in $\sigma ^2$ faster than the smallest equilibrium and slower than the largest equilibrium, so the profit maximizer is in the equilibrium set if $\sigma ^2$ is sufficiently high.

	\begin{figure}[tp]
		\centering
		\begin{tikzpicture}
			\begin{axis}[
				declare function={
					xnaught=15;
					m=-0.7;
					g=0.3333 ;
					done=2;
					dtwo=1;
					potsymmcutoff=-m*(done-dtwo)/(2*g);
					eqsymmcutoff=-m*(done-dtwo)/g;
					poteqcutoff = m*((done - 3*dtwo)*g - (done - dtwo)) / (g*(1 - 2*g));
				},
				ytick = {\empty},
				yticklabels = {\empty},
				xlabel={$\sigma^2$},
				ylabel={Policy $p_i$},
				xtick={eqsymmcutoff, poteqcutoff, potsymmcutoff},
				xticklabels={$-\mu \frac{d_1-d_2}{g}$, $C$, \empty},
				legend pos= outer north east,
				domain=0:25,
				xmin=0,
				xmax=poteqcutoff+1.2,
				extra x ticks={potsymmcutoff},
				extra x tick labels={$-\mu \frac{d_1-d_2}{2g}$},
				extra x tick style={
					major tick length=1.15\baselineskip,
					tick align=outside,
				}
				]
				\pgfmathsetmacro{\g}{g}
				\pgfmathsetmacro{\dOne}{done}
				\pgfmathsetmacro{\dTwo}{dtwo}
				\pgfmathsetmacro{\m}{m}
				\pgfmathsetmacro{\Xzero}{xnaught}		
				
				\addplot[ultra thick, blue, domain = 0:eqsymmcutoff]
				{ (1/(1-\g*\g) * (\dOne/\m + \g*\dTwo/\m) - \Xzero/\m)
					- (1/(1+\g)) * x/(2*\m*\m) };
				\addlegendentry{$p^*_1$}
				
				\addplot[ultra thick, red, domain = 0:eqsymmcutoff]
				{ (1/(1-\g*\g) * (\dTwo/\m + \g*\dOne/\m) - \Xzero/\m)
					- ((1+2*\g)/(1+\g)) * x/(2*\m*\m) };
				\addlegendentry{$p^*_2$}
				
				\addplot[black, thick, dashed, name path = lower , domain = eqsymmcutoff:25]
				{ (\dTwo/(1-\g) - \Xzero)/\m
					+ x/((1-\g)*(-2*\m*\m)) };
				\addlegendentry{$L$}
				
				\addplot[black, thick, dashdotted, name path = upper, domain = eqsymmcutoff:25]
				{ (\dOne/(1-\g) - \Xzero)/\m
					+ ((1-2*\g)/(1-\g))*x/(-2*\m*\m) };
				\addlegendentry{$U$}
				
				\addplot[blue!70!green, ultra thick, densely dotted, domain = 0:potsymmcutoff]
				{ (1/(1-(2*\g)*(2*\g)) * (\dOne/\m + (2*\g)*\dTwo/\m) - \Xzero/\m)
					- (1/(1+(2*\g))) * x/(2*\m*\m) };
				\addlegendentry{$r^*_1$}
				
				\addplot[red!70!green, ultra thick, densely dotted, domain = 0:potsymmcutoff]
				{ (1/(1-(2*\g)*(2*\g)) * (\dTwo/\m + (2*\g)*\dOne/\m) - \Xzero/\m)
					-((1+4*\g)/(1+(2*\g))) * x/(2*\m*\m) };
				\addlegendentry{$r^*_2$}
				
				\addplot[green, ultra thick, densely dotted, name path = potential, domain = potsymmcutoff:25]
				{ ((\dOne+\dTwo)/(2*(1-2*\g)) - \Xzero)/\m
					+ x/(-2*\m*\m) };
				\addlegendentry{$r^*$}
				
				\addplot[blue!30, opacity=0.7] fill between[of=lower and upper, soft clip={domain=eqsymmcutoff:poteqcutoff}];
				\addplot[green!30, opacity=0.7] fill between[of=lower and upper, soft clip={domain=poteqcutoff:25}];
				
			\end{axis}
		\end{tikzpicture}
		\caption{There are multiple equilibria in $\Gamma $ if $\sigma ^2>-\mu\frac{d_1-d_2}{g}$. The functions $L$ and $U$ are the smallest and largest equilibrium policy, respectively, defined if there are multiple equilibria; the functions $p_2^*$ and $p^*_1$ are the equilibrium policies of the leader and the follower, respectively, defined if there exists a unique equilibrium. The functions $r_2^*$ and $r^*_1$ are the equilibrium policies, with double coordination motives, of the leader and the follower, respectively, defined if there exists a unique equilibrium in the modified game. The function $r$ is the policy in the potential maximizer of the modified game, defined if there are multiple equilibria in the modified game.}\protect\label{fig:org_policies}
\end{figure}
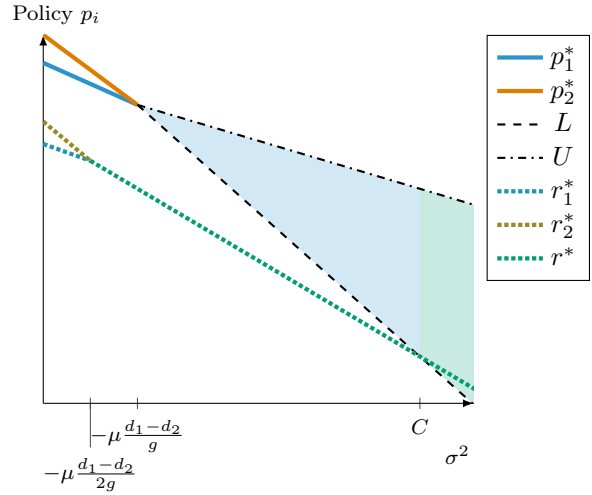

		A key reason for the presence of noise in the mapping from production processes to quantities is frictions along the chain of command. Suppose that each division manager communicates production instructions to lower‑level managers, who in turn interact with store managers, and so on. The division manager is uncertain about how her instructions are transmitted along this chain and how they are ultimately implemented. Complexity captures the noise perceived by the division manager: for example, a longer chain of communication makes the outcome of the original instructions less predictable.

	A necessary and sufficient condition to for maximization of total profits to be implemented in equilibrium is that complexity exceeds the threshold. 	Hence, centralized decision making is less desirable in the presence of coordination problems induced by complexity. Instead, the management of a holding company can leverage the coordination problems to make the maximization of total profits a focal point for subsidiary managers. This result complements the literature that studies informational asymmetries within
	organizations \citep{dessein_adaptive_2006, alonso2008, rantakari_governing_2008}. This model favors centralization because I abstracts from division managers’ private information.

	\newpage 
	\appendix
	\section{Proofs}
	The proofs of Lemma \ref{lem:ID}, Proposition \ref{prop:existence}, Proposition \ref{prop:equilibrium}, Corollary \ref{cor:multiplicity}, and Corollary \ref{cor:symmetriceq} are in the Supplemental appendix.

	\begin{proof}[Proof of Proposition \ref{prop:potentialmax}]
				Consider the function $v\colon\bm x\mapsto  \bm d ' \bm x -\frac{1}{2} \bm x'(\bm I- \bm G)\bm x$, and let $\bm G$ be symmetric. The games $\langle I, ((-\infty, X_0], u_i)_{i\in I}\rangle$ and $\langle I, ((-\infty, X_0], v)_{i\in I}\rangle$ are von-Neumann--Morgenstern equivalent \citep{morris_best_2004}. Thus, $\langle I, (P, U_i)_{i\in  I}\rangle$ and $\langle I, (P, V)_{i\in I}\rangle$ are von-Neumann--Morgenstern equivalent. Specifically, $2V$ is an exact potential \citep{monderer_potential_1996} for $\langle  I, (P, U_i)_{i\in I}\rangle$.
				
			Note that
		\begin{align*}
			V(\bm p) = \bm d ' E(\bm X(\bm p))- \frac{1}{2}  E(\bm X(\bm p)) '(\bm I- \bm G)E(\bm X(\bm p)) -\\ - \frac{1}{2}  \sum _{i\in I}\var (X_i(p_i))+\frac{1}{2}\sum_{(i, j)\in I^2} g_{ij}\cov (X_i(p_i), X_j(p_j)),
		\end{align*}
		and $V$ is strictly concave because $\bm I-\bm G$ is positive definite and $(p_i, p_j)\mapsto \cov (X_i(p_i), X_j(p_j))$ is concave. Letting $\partial V$ denote the superdifferential of $V$, by the rules of subdifferential calculus \citep{hiriart-urruty_fundamentals_2004}: $\bm y\in\partial V(\bm p)$ if and only if there exists a skew-complementary $\bm p$-consistent followership $\bm F$ such that
		\begin{align*}
		\text{for all} \ i\in I, \ 	y_i =  \mu d_i  -\mu E(X( p_i))+ \mu \sum_{j\in I}g_{ij}E(X(p_j))-\frac{1}{2}\sigma ^2 + \sum_{j\in I}g_{ij}\sigma ^2f_{ij}.
		\end{align*}
	For  all $p\ge p_0$, $\arg \max _{\bm p\in [p_0, p]^n}V(\bm p)\ne \emptyset $; moreover, $\bm p$ maximizes $V$ on $[p_0, p]^n$ if and only if there exists a $\bm p$-consistent skew-complementary followership $\bm F$ such that
		\begin{align*}
			\text{for all} \ i\in I, \ 	p_i&=p_0 \ \text{implies}	\  X_0\le \sum_{j\in I}m_{ij}(d_j +k-2\sum_{\ell\in I} g_{j\ell }kf_{j\ell}),\\
			p_i&\in (p_0, p)\ \text{implies}	\  E( X_i( p_i))=\sum_{j\in I}m_{ij}(d_j +k-2\sum_{\ell\in I} g_{j\ell }kf_{j\ell}), \\ 
			p_i&=p \ \text{implies}	\ E( X_i( p))\ge \sum_{j\in I}m_{ij}(d_j +k-2\sum_{\ell\in I} g_{j\ell }kf_{j\ell}).
		\end{align*}

		Finally, define $\overline x=\min_{i\in  I}\min_{\bm F\in[0, 1]^{n\times n}} \sum_{j\in I}m_{ij}(d_j +k-2\sum_{\ell\in I} g_{j\ell }kf_{j\ell})$ and $\overline p = p_0+\frac{1}{-\mu} (X_0-\overline x)_+$.  A potential maximizer exists in $P^n$ because $ \argmax_{\bm p\in [p_0, p]^n}V(\bm p)=\argmax_{\bm p\in P^n}V(\bm p)$ for all $p\ge\overline p$.
	\end{proof}

	\begin{proof}[Proof of Proposition \ref{prop:org}]
		Suppose $n=2$, $\bm G$ is symmetric, $\hat g\coloneqq g_{12}$, $\bm I-\bm G$ be invertible, and define 	\begin{align*}
			r^*\coloneqq	& 	\frac{d_1+d_2}{2(1-2\hat g)}\frac{1}{\mu} -\frac{X_0}{\mu}-\frac{\sigma ^2}{2\mu^2} \\
			L\coloneqq	&  \frac{1}{1-\hat g} \frac{d_2}{\mu} - \frac{X_0}{\mu} -\frac{1}{1-\hat g}\frac{\sigma ^2}{2\mu^2}  \\
			U\coloneqq & \frac{1}{1-\hat g} \frac{d_1}{\mu} - \frac{X_0}{\mu} -\frac{1-2 \hat g}{1- \hat g}\frac{\sigma ^2}{2\mu^2} .
		\end{align*}
		
	By Proposition \ref{prop:equilibrium}, the set of equilibria $\mathcal E$ is
	\begin{align*}
		\mathcal E \coloneqq \begin{cases}
			\left \{\begin{pmatrix}
				\frac{1}{1-{\hat g}^2}\Big ( \frac{d_1}{\mu } +{\hat g}\frac{d_2}{\mu }     \Big ) -\frac{X_0}{\mu } +\frac{1}{1+{\hat g}}\frac{k}{\mu}\\  \frac{1}{1-{\hat g}^2}\Big ( \frac{d_2}{\mu } +{\hat g}\frac{d_1}{\mu }     \Big ) -\frac{X_0}{\mu } + \frac{1+2{\hat g} }{1+{\hat g} }\frac{k}{\mu}
			\end{pmatrix}\right \}, & \text{if} \ d_1\ge d_2 +  2{\hat g}k,\\
			\left \{\begin{pmatrix}
				s\\s
			\end{pmatrix} | s\in [L, U]\right \}, & \text{if} \  d_1< d_2 +  2{\hat g} k,\\
		\end{cases}
	\end{align*}
 	if all equilibria are interior.
	By Proposition \ref{prop:potentialmax}, the potential maximizer is
	\begin{align*}
		\bm t  \coloneqq \begin{cases}
			\begin{pmatrix}
				\frac{1}{1-(2{\hat g} )^2}\Big ( \frac{d_1}{\mu } +(2{\hat g} )\frac{d_2}{\mu }     \Big ) -\frac{X_0}{\mu } +\frac{1}{1+(2{\hat g})}\frac{k}{\mu}\\  \frac{1}{1-(2{\hat g})^2}\Big ( \frac{d_2}{\mu } +(2{\hat g})\frac{d_1}{\mu }     \Big ) -\frac{X_0}{\mu } + \frac{1+4{\hat g}}{1+(2{\hat g})}\frac{k}{\mu}
			\end{pmatrix}, & \text{if} \ d_1\ge d_2 +  4{\hat g}k,\\
			\bm r^*, & \text{if} \  d_1< d_2 +  4{\hat g}k,
		\end{cases}
	\end{align*}
	if interior. A sufficient condition for interiority of the potential maximizer and all equilibria is
	\begin{align*}
		X_0 > \frac{d_1+2{\hat g}d_2+k}{1-2{\hat g}}.
	\end{align*}
	
	Observe that the potential is the utilitarian welfare of the game obtained from $\Gamma$ only by replacing $g_{12}$ with $2g_{12}$ \citep{jackson_games_2015}. Without loss, consider the case $d_1\ge d_2$. The utilitarian-maximal profile belongs to the equilibrium set if and only if
	\begin{align*}
		k\ge \hat C \;=\; \frac{\mu\bigl(d_1 g - d_1 - 3 d_2 g + d_2\bigr)}{{\hat g}(1-2{\hat g})}. \ \text{and} \ d_1\le  d_2+2k{\hat g}.
	\end{align*}

	The result follows from three observations. First, the production game admits a potential with a maximizer found by Proposition \ref{prop:potentialmax}, 	if  $2d_i=ba_i-bc_i$ and $bg=2{\hat g}$. Second, $W$ is the potential of the production game in which cost externalities are doubled. Third, the equilibrium set of the production game is characterized using Proposition \ref{prop:equilibrium}, because the two games are best-response equivalent \citep{morris_best_2004}.
	\end{proof}
	 \newpage 
	 
	\section{Supplemental appendix}
	
	Proposition \ref{prop:existence} and \ref{prop:equilibrium} are implied by Proposition \ref{app:prop:equilibriumexistence}, and Lemma \ref{lem:ID} is implied by Proposition \ref{app:prop:ID}. I establish these results in a model that nests the model in the main text as a special case.

	\subsection{Model}
	
		We consider an $n$-dimensional Gaussian process on $P\coloneqq [p_0, \infty)$, whose $i$-th dimension is denoted by $X_i$, and with the following restrictions, for parameters $(\mu_1, \dots, \mu_n)\in (-\infty, 0) ^n$ and $(\sigma _{ij})_{(i, j)\in I^2 }\in\mathbb R^{2n}_{+}$. The mean of the $i$-th component is given by $E(X_i(\cdot))\colon p_i\mapsto  X_0+\mu_i (p_i-p_0)$, the variance of the $i$-th component is given by $\var (X_i(\cdot))\colon p_i\mapsto \sigma_{ii}(p_i-p_0)$, and the covariances are determined by $\cov (X_i(p_i), X_j(p_j)) = \sigma_{ij}\min\{p_i-p_0, p_j-p_0\}$, for all $(i, j)\in I\times  I\setminus \{i\}$.
	
	We use the following notation: $k_{ij}=-\sigma_{ij}/(2\mu_i)$. The vector-valued function $(p_1, \dots, p_n) \mapsto (X_1(p_1), \dots, X_n(p_n))$ on $P^n$ evaluated at $\bm p$ is denoted by $\bm X(\bm p)$. The profile with $n$ outputs $\bm y\in\mathbb R^n$ is also denoted by $(y_1, \dots, y_n)$ and $(y_i, \bm y_{-i})$, the associated column vector is also denoted by $\bm y$. The $n$-$n$ matrix with $g_{ij}$ as its $i$-$j$th entry is denoted by $\bm G$. We denote by $\bm M$ the matrix $(\bm I-\bm G)^{-1}$, if well defined, by $v_{ij}$ and $\bm v _i$ the $i$-$j$ entry and the column vector corresponding to the $i$th row of a matrix $\bm V$, and by $ v_i$ the $i$-th entry of the vector $\bm v$. 
	
	We study the normal-form game $\Gamma = \langle  I, (P, U_i)_{i\in I} \rangle $, in which the payoff of player $i$ is given by $U_i\colon (p_1, \dots, p_n) \mapsto   E (u_i(X_1(p_1), \dots, X_n(p_n)))$. We make reference to the game with bounded strategy spaces $\Gamma_p \coloneqq \langle  I, ([p_0, p], U_i)_{i\in I} \rangle $, for $p\in (p_0, \infty)$, and the game without uncertainty $\Gamma^0\coloneqq \langle  I, (P, u_i(E(\bm X(\cdot))))_{i\in I} \rangle$.

	Given a player $i$, an \emph{$i$-followership} is a vector $\bm f\in [0,1]^n$; given a policy profile $\bm p$ , an $i$-followership $\bm f$ is \emph{$\bm p$-consistent} if: $p_i<p_j$ implies $f_{j}=1$ and $p_i>p_j$ implies $f_{j}=0$. A \emph{followership} is a matrix $\bm F\in[0,1]^{n\times n}$ such that the $i$-th row of $\bm F$ is an $i$-followership. Given a policy profile $\bm p$, a followership $\bm F$ is \emph{$\bm p$-consistent} if: the $i$-th row of $\bm F$ is $\bm p$-consistent for every $i\in I$. A followership $\bm F$ is \emph{skew-complementary} if: $f_{ij}=1-f_{ji}$ for all $(i, j)\in I^2$.

	\subsection{Proof of Proposition \ref{prop:existence} and \ref{prop:equilibrium}}

	\begin{lemma}\label{app:lem:quasiconcavity}
		For all $i\in I$, the payoff $p_i\mapsto U_i((p_i, \bm p_{-i}))$ is strictly quasiconcave. Moreover, if $\bm G\circ  \bm \Sigma $ is nonnegative, then $p_i\mapsto U_i((p_i, \bm p_{-i}))$ is strictly concave.
	\end{lemma}
	
	\begin{proof}
		Let $i\in I$, $\bm p\in P^n$, $N^+\coloneqq \{j\in I : g_{ij}\sigma_{ij}\ge 0\}$, and $N^-\coloneqq \{j\in I : g_{ij}\sigma_{ij}<0\}$. The payoff of player $i$ from $\bm p$ equals
	\begin{align*}
	U_i(\bm p) =	u_i(E(\bm  X(\bm  p))) - \var(X(p_i)) - \var (\sum _{j\in I} g_{ij}X(p_j)) \\ + 2\sum _{j\in I} g_{ij} \sigma_{ij} \min\{p_i-p_0, p_j-p_0\}.
	\end{align*}
	We have that $\min\{p_i-p_0, p_j-p_0\}$ is a concave and nondecreasing function of $p_i$. Hence, we have that $2 \sum_{j\in N^+} g_{ij}\sigma_{ij}\min\{p_i-p_0, p_j-p_0\}$ is a concave function of $p_i$. Therefore, $	u_i(E(\bm  X(\bm  p))) - \var(X_i(p_i)) - \var (\sum _{j\in I} g_{ij}X_j(p_j)) + 2\sum _{j\in N^+} g_{ij}\sigma_{ij}\min\{p_i-p_0, p_j-p_0\}$ is a strictly concave function of $p_i$, because $u_i(E(\bm  X(\bm  p))) - \var(X_i(p_i)) - \var (\sum _{j\in I} g_{ij}X_j(p_j))$ is a quadratic function of $p_i$ with second-order derivative at $p_i$ given by $-2\mu_i^2$.	Moreover, $2 \sum_{j\in N^-} g_{ij} \sigma_{ij}\min\{p_i-p_0, p_j-p_0\}$ is a nonincreasing function of $p_i$. Hence, $2 \sum_{j\in N^-} g_{ij}\sigma_{ij} \min\{p_i-p_0, p_j-p_0\}$ is quasiconcave as a function of $p_i$. Therefore, $p_i\mapsto U_i(\bm p)$ is strictly quasiconcave because it is the sum of a quasiconcave function and a strictly concave function. 
	\end{proof}
	
	\begin{lemma}\label{app:lem:equilibriumcharacterization}
		The policy $p_i$ is a best response to $\bm p_{-i}\in P^{n-1}$ if and only if there exists a $(p_i, \bm  p_{-i})$-consistent $i$-followership $\bm f_i$ such that: 
		\begin{align*}
			E(X_i(p_i))\le  d_i +\sum _{j\in I} g_{ij}E( X_j(p_j))  + k_{ii}-2 \sum _{j\in I}   g_{ij}k_{ij}f_{ij},
		\end{align*}
		with equality if $p_i>p_0$. 
	\end{lemma}
	\begin{proof}
		 By Lemma \ref{app:lem:quasiconcavity}, $p_i$ is a best response to $\bm p_{-i}$ if and only if there exists a $\bm p$-consistent $i$-followership $\bm f_i$ such that
		\begin{align*}
			-2\mu_i\big (  E (X_i(p_i))- d_i -\sum _{j\in I} g_{ij} E(X_j(p_j))\big )-\sigma _{ii} + 2\sum _{j\in I}   g_{ij}\sigma_{ij}f_{ij}\le 0,
		\end{align*}
		with equality if $p_i>p_0$. Hence, $p_i$ is a best response to $\bm p_{-i}$ if and only if there exists a $\bm p$-consistent $i$-followership $\bm f_i$ such that
		\begin{align*}
			  E (X_i(p_i))\le d_i +\sum _{j\in I} g_{ij} E(X_j(p_j))+k_{ii}- 2\sum _{j\in I}   g_{ij} k_{ij} f_{ij},
		\end{align*}
		with equality if $p_i>p_0$.
	\end{proof}
	
	\begin{proposition}\label{app:prop:equilibriumexistence}
	Let $\bm I-\bm G$ be invertible. There exists an equilibrium. Moreover, the policy profile $\bm p$ is an equilibrium if and only if there exists a $\bm p$-consistent followership $\bm F$ such that:
	\begin{align*}
		E(\bm X(\bm p))\le \bm M(\bm d +\operatorname{diag}(\bm K)-2(\bm G \circ \bm K \circ \bm F)\bm 1),
	\end{align*}
	and, for all $i\in I$, $p_i>p_0$ only if $E( X_i( p_i))=\sum_{j\in I}m_{ij}(d_j +k_{jj}-2\sum_{\ell\in I} g_{j\ell }k_{j\ell}f_{j\ell})$.
	\end{proposition}
	\begin{proof}
	 Let $\bm I - \bm G$ be invertible. By Lemma \ref{app:lem:equilibriumcharacterization}, the policy profile $\bm p$ is an equilibrium if and only if there exists a $\bm p$-consistent followership $\bm F$ such that:
	 \begin{align*}
	 	E(\bm X(\bm p))\le \bm M(\bm d +\operatorname{diag}(\bm K)-2(\bm G \circ \bm K \circ \bm F)\bm 1),
	 \end{align*}
	 and, for all $i\in I$, $p_i>p_0$ only if $E( X_i( p_i))=\sum_{j\in I}m_{ij}(d_j +k_{jj}-2\sum_{\ell\in I} g_{j\ell }k_{j\ell}f_{j\ell})$.
	For all $p\ge 0$, there exists a Nash equilibrium in the game $\langle  I, (U_i, [p_0, p])_{i\in I})\rangle$ by known arguments applying Brouwer's fixed point theorem and Lemma \ref{app:lem:quasiconcavity}; moreover, by Lemma  \ref{app:lem:quasiconcavity}, $\bm p$ is a Nash equlibrium of the game $\Gamma_p$ if and only if: there exists a $\bm p$-consistent followership $\bm F$ such that, for all $i\in I$,
	\begin{align*}
	p_i&=p_0 \ \text{implies}	\  X_0\le \sum_{j\in I}m_{ij}(d_j +k_{jj}-2\sum_{\ell\in I} g_{j\ell }k_{j\ell}f_{j\ell}),\\
		p_i&\in (p_0, p)\ \text{implies}	\  E( X_i( p_i))=\sum_{j\in I}m_{ij}(d_j +k_{jj}-2\sum_{\ell\in I} g_{j\ell }k_{j\ell}f_{j\ell}),\\
	p_i&=p \ \text{implies}	\ E( X_i( p))\ge \sum_{j\in I}m_{ij}(d_j +k_{jj}-2\sum_{\ell\in I} g_{j\ell }k_{j\ell}f_{j\ell}).
	\end{align*}
	Let's define $\overline x=\min_{i\in  I}\min_{\bm F\in[0, 1]^{n\times n}} \sum_{j\in I}m_{ij}(d_j +k_{jj}-2\sum_{\ell\in I} g_{j\ell }k_{j\ell}f_{j\ell})$ and $\overline p = p_0+\frac{1}{-\mu} (X_0-\overline x)_+$. For all $p\ge\overline p$, the set of Nash equilibria of $\langle  I, (U_i, [p_0, p])_{i\in I})\rangle$ is equal to the set of Nash equilibria of $\langle  I, (U_i, P)_{i\in I})\rangle$. Therefore, an equilibrium exists. 
	\end{proof}

	\subsection{Proof of Lemma \ref{lem:ID}}
	In this section, we consider $\bm X$ as an $n$-dimensional Gaussian process, with mean function $\bm \mu $ and covariance function $\bm \Sigma$ \citep{shreve_stochastic_2004}. We assume that there exists a 1-dimensional Gaussian process with covariance function $\overline k$ and a positive-definite matrix $\bm S$ such that $\sigma_{ij} (p_i, p_j)= s_{ij}\overline k(p_i, p_j)$; this construction is known as intrinsic correlation in statistics. We let $k_{ij} = s_{ij}/(-2\mu)$. Lemma \ref{lem:ID} is implied by Proposition  \ref{app:prop:ID}.
	\begin{proposition}\label{app:prop:ID}
		For all $i\in I$, let $g_{ij}\ge0$. Then, $U_i$ satisfies increasing differences if $(p_i, p_j)\mapsto E(X_i(p_i)X_j(p_j))$ satisfies increasing differences for all $j\in I$.
	\end{proposition}
	\begin{proof}
		Let's fix $(p_1', p_1, p_2', p_2)\in P^4$ with $p_1'>p_1$ and $p_2'>p_2$, and $(p_3, \dots, p_n)\in P^{n-2}$. Letting $d =U_1(p_1', p_2', \dots ) - U_1(p_1,p_2', \dots ) - U_1(p_1', p_2, \dots ) + U_1(p_1, p_2, \dots)$, it holds that
		\begin{align*}
			d & =  2g_{12}E(X_1(p_1')-X_1(p_1))E( X_2(p_2')-X_2(p_2) )+\\ &+2g_{12}\cov(X_1(p_1')-X_1(p_1), X_2(p_2')-X_2(p_2))\\
			&= 2g_{12} E((X_1(p_1')-X_1(p_1))( X_2(p_2')-X_2(p_2) )).
		\end{align*}
		Thus, for fixed $(p_3, \dots, p_n)\in P^{n-2}$, we have that:  $d \ge 0$ for all $(p_1', p_1, p_2', p_2)\in P$ if and only if $E(X_1(p_1)X_2(p_2))$ exhibits increasing differences in $(p_1, p_2)$ on $P^2$.
	\end{proof}
	
	\begin{remark}
		An example of a covariance function $\overline k$ that implies that $U_i$ satisfies increasing differences and generalizes the Brownian motion has $k(p_i, p_j) = \int _{[0,1]} \frac{((p_i-u)_+)^{m-1}((p_j-u)_+)^{m-1}}{((m-1)!)}\diff u$. If $U_i$ satisfies increasing differences, then there exists a Nash equilibrium in the game $\Gamma _p$, for all $p\in P$, by known results \cite{milgrom_rationalizability_1990}.
		
		If the mean function has the same monotonicity in every direction, then a sufficient condition for $(d \ge 0 \iff g_{12}\ge 0)$ is that the function $(p_1, p_2)\mapsto \cov (X_1(p_1), X_2(p_2))$ exhibits increasing differences. Let the function $(p_1, p_2)\mapsto E(X_1(p_1)X_2(p_2))$ exhibit increasing differences. It holds that
		\begin{align*}
			d & \ge  2g_{12}E(X_1(p_1')-X_1(p_1))E( X_2(p_2')-X_2(p_2) ) \\
			&	\ge u_1(E(\bm X(p_1', p_2', \dots ))) - u_1(E(\bm X(p_1,p_2', \dots ))) \\ & - u_1(E(\bm X(p_1', p_2, \dots ))) + u_1(E(\bm X(p_1, p_2, \dots))),
		\end{align*}
		and the inequality is strict whenever $g_{12}>0$ and $\cov(X_1(p_1) X_2(p_2))$ exhibits strictly increasing differences.
	\end{remark}

	\subsection{Proof of Corollary \ref{cor:symmetriceq} and \ref{cor:multiplicity}}
	\begin{proof}[Proof of Corollary \ref{cor:symmetriceq} ]
		By computation of $(\bm I - \bm G)^{-1}$, we have that the diagonal element is $\frac{1-g(n-1) +g}{(1-g(n-1))(1+g)}$ and the off-diagonal element is: $\frac{g}{(1-g(n-1))(1+g)}$, so that:
		\begin{align*}
			(\bm I-\bm G)^{-1}_{ii}  - (\bm I-\bm G)^{-1}_{im} = \frac{1}{1+g}.
		\end{align*}
			We have that, for every $i, m\in N$ and equilibrium $\bm p$, by the preceding equality and Proposition \ref{prop:equilibrium},
		\begin{align*}
			 E X (p_i) -		 E X (p_m)  & =  \beta_i - \beta_m  +\frac{g}{1+g}\left (  \sum _{\ell \in I}a_{i\ell}  - \sum _{\ell \in I}a_{m\ell}  \right )k\\
			& =   \beta_i - \beta_m  -\frac{g}{1+g}2k +\frac{g}{1+g} \left (  \sum _{\ell \in I\setminus \{i, m\} }a_{i\ell}  -a_{m\ell}  \right )k.
		\end{align*}	
	\end{proof}	

	\begin{proof}[Proof of Corollary \ref{cor:multiplicity} ]
By Lemma \ref{lem:ID}, there exist a greatest and a least equilibrium, respectively $\bm q$ and $\bm s$ \citep{vives_nash_1990, milgrom_rationalizability_1990}. It suffices to establish that:
	\begin{align*}
		\mathbb E \bm X (\bm q) =\bm \beta +\bm uk \quad \text{and} \quad 
		\mathbb E \bm X (\bm s) = \bm \beta +\bm 2k - \bm uk.
	\end{align*}

	Let $([p_i<p_j], \  i,j\in  N)$ and $([p_i\le p_j], \  i,j\in  N)$  be two $n$-by-$n$ matrices. We define $\bm \Gamma_+(p) = \bm G \circ ([p_i<p_j], \  i,j\in  I)$ and $\bm \Gamma_-(p) = \bm G \circ ([p_i\le p_j], \  i,j\in  I)$. By Proposition \ref{prop:equilibrium}, the interior policy profile
	$ p$ is an equilibrium if and only if:
	\begin{align*}
		k(\bm I -2 \bm   \Gamma_-(p))\bm 1\le (\bm I-  \bm G) \left (  E \bm  X (p) - \bm  \beta \right ) \le k(\bm I-2 \bm \Gamma_+(p))\bm 1.
	\end{align*}
	Hence, for interior equilibrium $\bm p$, we have
	\begin{align*}
		\bm  \beta + (\bm I- \bm G )^{-1}(\bm I -2 \bm   \Gamma_-(p))\bm 1k & = \bm \beta +2\bm 1k - \bm uk, \\
		\bm  \beta + (\bm I- \bm G )^{-1} (\bm I-2 \bm \Gamma_+(p))\bm 1 k& =  \bm \beta +\bm 1k.
	\end{align*}
	The result follows.	
	\end{proof}	
	\newpage 

	\bibliographystyle{te}
	\addcontentsline{toc}{section}{\refname}\bibliography{jmpbibl.bib}

\begin{thebibliography}{28}
\newcommand{\enquote}[1]{``#1''}
\providecommand{\natexlab}[1]{#1}
\providecommand{\url}[1]{\texttt{#1}}
\providecommand{\urlprefix}{URL }
\providecommand{\bibAnnoteFile}[1]{%
  \IfFileExists{#1}{\begin{quotation}\noindent\textsc{Key:} #1\\
  \textsc{Annotation:}\ \input{#1}\end{quotation}}{}}
\providecommand{\bibAnnote}[2]{%
  \begin{quotation}\noindent\textsc{Key:} #1\\
  \textsc{Annotation:}\ #2\end{quotation}}

\bibitem[{Alonso et~al.(2008)Alonso, Dessein, and Matouschek}]{alonso2008}
Alonso, Ricardo, Wouter Dessein, and Niko Matouschek (2008), \enquote{When does
  coordination require centralization?} \emph{American Economic Review}, 98,
  145--79.
\bibAnnoteFile{alonso2008}

\bibitem[{Alonso et~al.(2015)Alonso, Dessein, and Matouschek}]{alonso2015}
Alonso, Ricardo, Wouter Dessein, and Niko Matouschek (2015),
  \enquote{Organizing to adapt and compete.} \emph{American Economic Journal:
  Microeconomics}, 7, 158--87.
\bibAnnoteFile{alonso2015}

\bibitem[{Asch(1951)}]{asch_effects_1951}
Asch, {S.E.} (1951), \enquote{Effects of group pressure upon the modification
  and distortion of judgments.} In \emph{Groups, leadership and men; research
  in human relations.}, 177--190, Carnegie Press, Oxford, England.
\bibAnnoteFile{asch_effects_1951}

\bibitem[{Ballester et~al.(2006)Ballester, Calvó-Armengol, and
  Zenou}]{ballester_who_2006}
Ballester, Coralio, Antoni Calvó-Armengol, and Yves Zenou (2006),
  \enquote{Who's who in networks. wanted: The key player.} \emph{Econometrica},
  74, 1403--1417.
\bibAnnoteFile{ballester_who_2006}

\bibitem[{Bardhi(2024)}]{bardhi_selective_2023}
Bardhi, Arjada (2024), \enquote{Attributes: Selective learning and influence.}
  \emph{Econometrica}, 92, 311--353.
\bibAnnoteFile{bardhi_selective_2023}

\bibitem[{Bardhi and Bobkova(2023)}]{bardhi_local_2023}
Bardhi, Arjada and Nina Bobkova (2023), \enquote{Local evidence and diversity
  in minipublics.} \emph{Journal of Political Economy}, 131, 2451--2508.
\bibAnnoteFile{bardhi_local_2023}

\bibitem[{Bardhi and Callander(2026)}]{bardhi_learning_2026}
Bardhi, Arjada and Steven Callander (2026), \enquote{Learning in a correlated
  world.} \emph{Annual Review of Economics}, 00, 00. DOI:
  https://doi.org/10.1146/annurev-economics-051624-072515.
\bibAnnoteFile{bardhi_learning_2026}

\bibitem[{Callander(2011)}]{callander_searching_2011}
Callander, Steven (2011), \enquote{Searching and learning by trial and error.}
  \emph{American Economic Review}, 101, 2277--2308.
\bibAnnoteFile{callander_searching_2011}

\bibitem[{Cetemen et~al.(2023)Cetemen, Urgun, and
  Yariv}]{cetemen_collective_2023}
Cetemen, Doruk, Can Urgun, and Leeat Yariv (2023), \enquote{Collective
  progress: Dynamics of exit waves.} \emph{Journal of Political Economy}, 131,
  2402--2450.
\bibAnnoteFile{cetemen_collective_2023}

\bibitem[{Dessein and Santos(2006)}]{dessein_adaptive_2006}
Dessein, Wouter and Tano Santos (2006), \enquote{Adaptive organizations.}
  \emph{Journal of Political Economy}, 114, 956--995.
\bibAnnoteFile{dessein_adaptive_2006}

\bibitem[{Diamond(1982)}]{diamond_aggregate_1982}
Diamond, {Peter A.} (1982), \enquote{Aggregate demand management in search
  equilibrium.} \emph{Journal of Political Economy}, 90, 881--894.
\bibAnnoteFile{diamond_aggregate_1982}

\bibitem[{Garfagnini(2018)}]{garfagnini_uncertainty_2018}
Garfagnini, Umberto (2018), \enquote{{Uncertainty in a Connected World}.}
  Working Paper.
\bibAnnoteFile{garfagnini_uncertainty_2018}

\bibitem[{Gibbons et~al.(2013)Gibbons, Matouschek, and
  Roberts}]{gibbons_decisions_2013}
Gibbons, Robert, Niko Matouschek, and John Roberts (2013), \enquote{Decisions
  in organizations.} In \emph{The Handbook of Organizational Economics},
  373--431, Princeton University Press.
\bibAnnoteFile{gibbons_decisions_2013}

\bibitem[{Hiriart-Urruty and
  Lemaréchal(2004)}]{hiriart-urruty_fundamentals_2004}
Hiriart-Urruty, Jean-Baptiste and Claude Lemaréchal (2004), \emph{Fundamentals
  of convex analysis}. Springer Science \& Business Media.
\bibAnnoteFile{hiriart-urruty_fundamentals_2004}

\bibitem[{Jackson and Zenou(2015)}]{jackson_games_2015}
Jackson, Matthew~O. and Yves Zenou (2015), \enquote{Games on networks.} In
  \emph{Handbook of Game Theory with Economic Applications} (H.~Peyton Young
  and Shmuel Zamir, eds.), volume~4, 95--163, Elsevier.
\bibAnnoteFile{jackson_games_2015}

\bibitem[{Jørgensen and Zaccour(2014)}]{jorgensen_survey_2014}
Jørgensen, Steffen and Georges Zaccour (2014), \enquote{A survey of
  game-theoretic models of cooperative advertising.} \emph{European Journal of
  Operational Research}, 237, 1--14.
\bibAnnoteFile{jorgensen_survey_2014}

\bibitem[{Krech et~al.(1962)Krech, Crutchfield, and
  Ballachey}]{krech_individual_1962}
Krech, David, {Richard S}. Crutchfield, and {Egerton L.} Ballachey (1962),
  \emph{Individual in society: {A} textbook of social psychology.} McGraw-Hill.
\bibAnnoteFile{krech_individual_1962}

\bibitem[{Lin(2023)}]{lin_strategic_2023}
Lin, Jianjing (2023), \enquote{Strategic complements or substitutes? the case
  of adopting health information technology by u.s.~hospitals.} \emph{The
  Review of Economics and Statistics}, 105, 1237--1254.
\bibAnnoteFile{lin_strategic_2023}

\bibitem[{Marschak and Radner(1972)}]{marschak_economic_1972}
Marschak, Jacob and Roy Radner (1972), \emph{Economic Theory of Teams}.
  Number~22 in Cowles Foundation for Research in Economics Monographs, Yale
  University Press.
\bibAnnoteFile{marschak_economic_1972}

\bibitem[{Milgrom and Roberts(1990)}]{milgrom_rationalizability_1990}
Milgrom, Paul and John Roberts (1990), \enquote{Rationalizability, learning,
  and equilibrium in games with strategic complementarities.}
  \emph{Econometrica}, 58, 1255--1277.
\bibAnnoteFile{milgrom_rationalizability_1990}

\bibitem[{Monderer and Shapley(1996)}]{monderer_potential_1996}
Monderer, Dov and { Lloyd S.} Shapley (1996), \enquote{Potential {Games}.}
  \emph{Games and Economic Behavior}, 14, 124--143.
\bibAnnoteFile{monderer_potential_1996}

\bibitem[{Morris and Ui(2004)}]{morris_best_2004}
Morris, Stephen and Takashi Ui (2004), \enquote{Best response equivalence.}
  \emph{Games and Economic Behavior}, 49, 260--287.
\bibAnnoteFile{morris_best_2004}

\bibitem[{Neyman(1997)}]{neyman_correlated_1997}
Neyman, Abraham (1997), \enquote{Correlated equilibrium and potential games.}
  \emph{International Journal of Game Theory}, 26, 223--227.
\bibAnnoteFile{neyman_correlated_1997}

\bibitem[{Radner(1962)}]{radner_team_1962}
Radner, Roy (1962), \enquote{Team decision problems.} \emph{The Annals of
  Mathematical Statistics}, 33, 857--881.
\bibAnnoteFile{radner_team_1962}

\bibitem[{Rantakari(2008)}]{rantakari_governing_2008}
Rantakari, Heikki (2008), \enquote{Governing adaptation.} \emph{The Review of
  Economic Studies}, 75, 1257--1285.
\bibAnnoteFile{rantakari_governing_2008}

\bibitem[{Shreve(2004)}]{shreve_stochastic_2004}
Shreve, Steven~E (2004), \emph{Stochastic calculus for finance 2,
  Continuous-time models}. Springer, New York, NY; Heidelberg.
\bibAnnoteFile{shreve_stochastic_2004}

\bibitem[{Topkis(1998)}]{topkis_supermodularity_1998}
Topkis, {Donald M.} (1998), \emph{Supermodularity and Complementarity}.
  Princeton University Press.
\bibAnnoteFile{topkis_supermodularity_1998}

\bibitem[{Vives(1990)}]{vives_nash_1990}
Vives, Xavier (1990), \enquote{Nash equilibrium with strategic
  complementarities.} \emph{Journal of Mathematical Economics}, 19, 305--321.
\bibAnnoteFile{vives_nash_1990}

\end{thebibliography}
	
\end{document}